\documentclass{ecai}  % use option [doubleblind] for double blind submission and hiding the authors section

\usepackage{graphicx}
\usepackage{latexsym}
\usepackage{algorithm}
\usepackage{algpseudocode}
\usepackage{amssymb}
\usepackage{bbm}
\usepackage{booktabs}
\usepackage[footnotesize]{caption}
\usepackage{color}
\usepackage{enumitem}
\usepackage[hidelinks]{hyperref}
\usepackage{graphicx}
\usepackage{graphics}
\usepackage[switch]{lineno}
\usepackage{mathtools}
\usepackage{multirow}
\usepackage{verbatim}
\usepackage{soul} %%%
\usepackage{subcaption}
\usepackage{url}
\usepackage{xcolor}

\newcommand{\LR}{\mathcal{L}}
\newcommand{\N}{\mathbb{\N}}
\newcommand{\R}{\mathbb{R}}

\newtheorem{lemma}[theorem]{Lemma}
\newtheorem{definition}{Definition}

\newenvironment{proof}{{\noindent \it Proof.}}{\hfill $\blacksquare$\par}

\begin{document}

\begin{frontmatter}

\title{Integrated Private Data Trading Systems \\for Data Marketplaces}

\author[A,B]{\fnms{Weidong}~\snm{Li}\thanks{Corresponding Author. Email: wli916@aucklanduni.ac.nz}}
\author[A]{\fnms{Mengxiao}~\snm{Zhang}\thanks{Corresponding Author. Email: mengxiao.zhang@uestc.edu.cn}}
\author[B]{\fnms{Libo}~\snm{Zhang}} 
\author[B]{\fnms{Jiamou}~\snm{Liu}}

\address[A]{School of Computer Science and Engineering, University of Electronic Science and Technology of China, China}
\address[B]{School of Computer Science, The University of Auckland, New Zealand}

\begin{abstract}
In the digital age, data is a valuable commodity, and data marketplaces offer lucrative opportunities for data owners to monetize their private data. However, data privacy is a significant concern, and differential privacy has become a popular solution to address this issue. Private data trading systems (PDQS) facilitate the trade of private data by determining which data owners to purchase data from, the amount of privacy purchased, and providing specific aggregation statistics while protecting the privacy of data owners. However, existing PDQS with separated procurement and query processes are prone to  over-perturbation of private data and lack trustworthiness. To address this issue, this paper proposes a framework for PDQS with an integrated procurement and query process to avoid excessive perturbation of private data. We also present two instances of this framework, one based on a greedy approach and another based on a neural network. Our experimental results show that both of our mechanisms outperformed the separately conducted procurement and query mechanism under the same budget regarding accuracy.
\end{abstract}

\end{frontmatter}

\section{Introduction}\label{sec:intro}

% motivate the data marketplace where data owners sell their private data to data brokers.
%In today's digital era, data has emerged as a highly valuable commodity. The advent of {\em data marketplaces} has opened up new avenues for {\em data owners} to monetise their private data, while simultaneously contributing to technological and scientific progresses. However, with data privacy becoming an increasingly pressing concern, one effective solution that has gained popularity is differential privacy \cite{dwork2006calibrating}.

In the modern digital era, data have become a highly valuable asset, providing opportunities to gain insights, make informed decisions, and drive innovation across various domains, leading to significant economic benefits for individuals, businesses, and governments \cite{manyika2011big}. 
However, this abundance of data also brings about a pressing concern: data privacy leakage. With recent data breaches affecting millions of $533$ Facebook users\footnote{\url{https://www.bbc.com/news/technology-56815478}}, $150,000$ NHS patients\footnote{\url{https://www.bbc.com/news/technology-44682369}} and fitness tracking app Strava users\footnote{\url{https://www.theguardian.com/world/2018/jan/28/fitness-tracking-app-gives-away-location-of-secret-us-army-bases}}, privacy preservation has become increasingly critical. This creates an inherent conflict between preserving privacy for individuals and unlocking the significant economic value for society.  {\em Data marketplaces} have emerged as a potential solution to this issue, as they compensate {\em data owners} for their private data to offset the potential privacy loss. By providing compensation, data owners may be more willing to share their data, allowing it to be further utilised \cite{zhang2022experimental}.
%The advent of {\em data marketplaces} proposes a solution to the conflicts between utilising (private) data for technological and scientific progress and addressing the privacy issues for {\em data owners} as a data marketplace offers a new avenue to monetise data owners' private data. % {\teal (M: add more details.)}

%The advent of {\em data marketplaces} has opened up new avenues for {\em data owners} to monetise their private data, while simultaneously contributing to technological and scientific progresses. However, with data privacy becoming an increasingly pressing concern, one effective solution that has gained popularity is differential privacy \cite{dwork2006calibrating}.

% the scenario
%{\red (M: Please fill up the empty citations.) }

Imagine a scenario where a data analyst, also referred to as a {\em data consumer}, wishes to obtain aggregation statistics from individual data owners through a data marketplace. These data owners are willing to share their private data on the data marketplace as long as they receive fair compensation. % and their privacy is safeguarded to a certain extent. 
Each data owner has a specific compensation amount in mind, known as their {\em valuation}, which they keep to themselves. Additionally, the data owners require their privacy to be safeguarded to a certain extent. To facilitate this transaction, a {\em data broker} acts as an intermediary between the data owners and data consumers on behalf of the data marketplace, procuring data from data owners and analysing the collected data for data consumers.

A {\em private data trading system} (PDQS) \cite{zhang2023smartauction} is a mechanism that enables the trade of private data, %identifies which data owners to buy data from and the amount of privacy to be traded, generates specific aggregation statistics, and ensures the protection of the privacy of data owners.
determining how much to compensate the data owners, generating specific aggregation statistics, and ensuring the privacy protection of the data owners.  
%
% Let us consider a data analyst, also referred to as a data consumer, who wants to obtain some aggregation statistics from individual data owners who are privacy-aware. These data owners may be willing to share their private data on the data marketplace if they receive fair compensation and their privacy is protected to some certain degree. Differential privacy (DP) can be adopted as a measure of the privacy loss for each data owner, indicating the amount of privacy compromised during the data trading. In the centralised setting of DP (CDP), data owners report their private data directly to the mechanism before perturbation, while in the local setting of DP (LDP), the private data is perturbed before sent to the mechanism. A mechanism that facilitates the trade of private data, determines which data owners to purchase data from and the amount of privacy purchased, provides specific aggregation statistics, and appropriately protects the privacy of data owners, is called a private data trading system (PDQS) \cite{}.
% existing works: GR
Starting with Ghosh and Roth's work \cite{ghosh2011selling}, a typical PDQS consists of two main components: a {\em procurement process} determines the selection of data owners from whom data will be purchased and the payment for the selected data owners; a {\em query process} executes a query on the procured data, adds noise to the data, and outputs query results.

%Most existing PDQS ensure CDP, with Ghosh and Roth's work \cite{ghosh2011selling} being a notable example, which has been extended by several works \cite{dandekar2014privacy}\cite{zhang2020selling}. FairQuery, proposed by Ghosh and Roth, considers binary private data, i.e., the private data is either 0 or 1, and provides an estimation of the summation of private data of all data owners, also known as a sum query. The data owners are required to report their private data and the valuations of the private data to the mechanism. The mechanism works as a data broker who acts as an intermediary between data owners and the data consumer, purchasing a subset of data owners, querying on the collected dataset, and providing the query answer with a random noise based on the consumer's budget and data owners' reported valuations.

% why does this kind of system fail? data broker has three issues, double perturbation

{\em Differential privacy} (DP) \cite{dwork2006calibrating}  is a privacy concept that can be employed to measure the degree of privacy loss for each data owner, indicating the amount of privacy compromised during the data trading. A typical method to achieve DP is to execute the query on the {\em raw} datasets and add noise to the true query answer. In this case, the data broker, who executes the query, is deemed to be trustworthy and has the full access to the raw data. However, in practice, the data broker may not be trustworthy: once obtaining the raw data, the data broker may resell the data to third parties for profit. 
Given that, a local model of DP, known as {\em local differential privacy} (LDP) \cite{duchi2013local} is proposed. In LDP model, each data owner adds noise to her private data locally before sharing it to the data broker. Most of the existing PDQS deploys the DP model \cite{ghosh2011selling,dandekar2014privacy,cummings2022optimal,zhang2023survey}. We instead extend the PDQS design to the LDP model. 

Existing works for implementing PDQS execute the procurement and the query processes one after another \cite{fallah2022optimal}. Specifically, the procurement process first selects a subset of data owners. Then to achieve LDP, the query process runs a {\em local randomiser} \cite{erlingsson2014rappor} where a smaller subset of data owners are selected randomly and submit their true data while each of the others submit a random number. In this way, the data owners who submit random numbers are paid with no contribution to the query in the query process.
In other words, the budget is not put to good use and the query accuracy could be improved. Therefore, a question arises: {\em How to design a PDQS with an integrated procurement and query process to address this issue? }

A challenge arises when designing such an integrated procurement and query process. The procurement process and the query process serve for different purposes and thus are expected to have different properties: the procurement process is used to incentivise the data owners to reveal their valuations for pricing purposes and it should satisfy incentive compatibility (IC), individual rationality (IR) and budget feasibility (BF) properties (see formal definitions in Sec.~\ref{sec:problem}), while the query process should satisfy LDP. Therefore, we need to design proper rules of the integrated process that meets all of these properties and thus can serve well for both procurement and query processes. 

To address the research question, we propose a new framework of the PDQS, named {\em Integrated PDQS}. The core of this framework is embedding the procurement process into the query process such that data owner selection and data perturbation happen simultaneously. To be specific, in the new framework, we first assign an allocation probability that a data owner would be procured, then we use this probability to obfuscate the data in the local randomiser. In this way, all selected data owners, once being paid, contribute to the query answer. See details in Sec.~\ref{sec:framework}. 

Then we propose two PDQS instantiations that implement the proposed framework. The first mechanism, {\bf g}reedy {\bf p}rivate {\bf q}uery {\bf m}echanism (GPQM), selects data owners with lower valuation in a greedy manner until the budget is exhausted. This mechanism allows us to use any non-increasing function as the allocation probability, but requires manual selection of the function and its parameters. Our second mechanism, {\bf n}eural-based {\bf p}rivate {\bf q}uery {\bf m}echanism (NPQM), parameterises the allocation probability with a neural network. The parameters are learned with the dual-ascent algorithm. We theoretically proved that both GPQM and NPQM satisfy desirable properties. See details in Sec.~\ref{sec: GPQM} and \ref{sec: NPQM}.  
We also empirically validate the performance of the propose two PDQS instances. Experimental results show that both of our proposed methods outperform the benchmarks, with NPQM performing better than GPQM.

The contributions of this paper are summarised as follows:
\begin{itemize}
    \item We propose a new framework of PDQS that integrates the procurement and query process. 
    \item We design two PDQS instances to implement the propose framework: GPQM is equipped with a non-decreasing function as the allocation probability, and NPQM uses neural network to learn an optimal allocation probability.
    \item We empirically show the strengths of the two proposed PDQSs. 
\end{itemize}

\section{Related works}
%Armstrong and Durfee first introduced the concept of an information marketplace\cite{armstrong1998mixing}. Jaisingh et al. proposed the notion of a data marketplace in their study on pricing personal data \cite{jaisingh2008privacy}. 

%The concept of information marketplace was first proposed by Armstrong and Durfee\cite{armstrong1998mixing}. Then, Jaisingh et al. studied the pricing of personal data, and proposed the notion of a data marketplace \cite{jaisingh2008privacy}.

%first study, their settings (CDP), how to modify, issues
%Regarding to the privacy issue in a data marketplace, 
In their seminal work, Ghosh and Roth \cite{ghosh2011selling} formalised the concept of trading private data under differential privacy with their proposed mechanism, FairQuery, which includes independent procurement and query processes. In the procurement process, a subset of data owners is selected based on their bids, and their private data is reported to the system without perturbation. In the query process, random noise is added to the query answer based on the collected dataset to guarantee DP. This leads to FairQuery and its extensions e.g. FairInnerProduct~\cite{dandekar2014privacy} heavily relying on a trusted data broker.

To tackle the distrust between data owners and the data broker, Wang et al.~\cite{wang2016buying} and Fallah et al.~\cite{fallah2022optimal} proposed PDQSs by enabling data owners to perturb their private data locally to ensure local DP. 
%Wang et al. \cite{wang2016buying} and Fallah et al. \cite{fallah2022optimal} proposed PDQSs that address the aforementioned issues{\color{magenta}What is the aforementioned issue?} by enabling data owners to perturb their private data locally to ensure local DP. 
However, they do not consider the budget constraints of the data consumer. Additionally, they rely on separated procurement and query processes, leading to excessive perturbation of private data.

There are other works that extended Ghosh and Roth's work from different perspectives, e.g., considering the correlations between data owners' values and valuations \cite{roth2012conducting, fleischer2012approximately, ligett2012take, nissim2014redrawing, chen2018optimal, chen2019prior}, the cases that data owners' private data is not verifiable \cite{ghosh2014buying}, the scenario where different levels of data accuracy are provided by various data brokers \cite{cummings2015accuracy}, the network effect on data owners' participation \cite{zhang2016privacy, jia2023incentivising}, single-minded data owners \cite{zhang2020selling}, data owners getting benefits from the statistic based on reported private data \cite{cummings2022optimal, fallah2022optimal}.
There are also studies working on the design of truthful mechanisms for trading data without preserving data under differential privacy \cite{cai2015optimum, liu2016learning}, privacy-aware mechanisms that preserve user bids \cite{mcsherry2007mechanism, nissim2012privacy, abernethy2019learning, lei2020privacy}, pricing mechanisms that charging users based on their perturbed private data \cite{chen2022differential}, and pricing strategies based on data quality \cite{yu2017data}.

To the best of our knowledge, there is no existing PDQS addressing our problem of trading private data under LDP while integrating procurement and query processes and ensuring IC, IR and BF properties (formally defined in Sec.~\ref{sec:problem}). Thus, we propose our framework to solve this problem and two instantiations of the framework.

\section{Problem formulation}\label{sec:problem}
% marketplace
Consider a data transaction with a data consumer and $n$ data owners. The data consumer aims to obtain some aggregation information about the data owners, e.g., how many people are infected in the population, which can be denoted as a {\em query} $f$. The data consumer has a budget, denoted by $\beta \in \R_+$, for the query. 

%{\red M: I think in the experiment, the domain is not $\{0,1\}$. If so, please revise it here. Also, fix Alg.~\ref{alg:LR} and Lemma~\ref{lemma:lr}. }

Each data owner $i$ has private data $t_i \in \{0, 1\}$. The data owner is willing to sell her private data to the data consumer, given reasonable compensation and privacy protection. Let $\varepsilon_i \in \R_+$ be a {\em privacy parameter}. %\footnote{The formal definition will be given soon.}.
When her data is used in an $\varepsilon_i$-differential privacy manner, she suffers a privacy cost $c_i \coloneqq \varepsilon_i \theta_i$, where $\theta_i \in [\underline{\theta}, \overline{\theta}]$ is her {\em valuation} to a unit of privacy. A data owner can be represented by a tuple $s_i \coloneqq (t_i, \theta_i)$. We use $\vec{s} \coloneqq (s_1, \ldots, s_n)$ to denote the data owners, $\vec{t}\coloneqq (t_1, \ldots, t_n)$ and $\vec{\theta} \coloneqq (\theta_1, \ldots, \theta_n)$ to denote the private data and the valuation vector of all $n$ data owners, respectively. Also, we use $\vec{\theta}_{-i} \coloneqq (\theta_1, \ldots, \theta_{i-1},\theta_{i+1},\ldots,\theta_n)$ to denote the valuation vector of all data owners but $s_i$. Data owners may misreport their valuations to, for example, gain more compensation. Let $b_i \in [\underline{\theta}, \overline{\theta}]$ denote the reported valuation of data owner $s_i$. Similarly, we have vectors $\vec{b} \coloneqq (b_1, \ldots, b_n)$ and $\vec{b}_{-i} \coloneqq (b_1, \ldots, b_{i-1},b_{i+1},\ldots, b_n)$.

We are committed to developing a Private Data Query System (PDQS) that takes the set of data owners and the data consumer's budget as inputs. The PDQS consists of four components: an \textit{allocation function} that determines the selection probabilities of data owners, a \textit{payment function} that calculates the compensation for each data owner, an \textit{LDP algorithm} that ensures privacy by adding noise locally, and a \textit{query function} that implements the query $f$ based on the collected data. Through the collaboration of these components, the PDQS efficiently processes queries while ensuring the privacy of data owners. 

Next, we formally define these functions and outline the desired properties they should possess.

%\medski

%We aim to establish a Private Data Query System (PDQS) that takes the set of data owners and the budget of the data consumer as the inputs, first decides whose data to purchase from and how much compensation to pay to each data owner, and then generates the query answer based on the collected data while preserving the privacy of data owners. The first step is referred to as the procurement process, and the second step is referred to as the query process. We introduce the procurement process and the query process individually.

%The {\em procurement process} decides the probability that each data owner is selected by the system and the compensation for each of them by applying an allocation function and a payment function. 

\begin{definition}
    An {\em allocation function} $q: [\underline{\theta}, \overline{\theta}]^n \to [0,1]^n$ is a mapping from the bids of all data owners to an allocation result $\vec{q}\coloneqq (q_1, \ldots, q_n)$, where $q_i$ is the probability that data owner $s_i$ is selected.  
\end{definition}
We also write $q_i(\vec{b})$ as the allocation function of an individual data owner $s_i$. When the context is clear, we write $q_i$ for short.

\begin{definition}
    A {\em payment function} $p: [\underline{\theta}, \overline{\theta}]^n \to \mathbb{R_+}^n$ is a mapping from the bids of all data owners, represented by payment vector $\vec{p}\coloneqq (p_1, \ldots, p_n)$, where $p_i$ is the compensation of data owner $s_i$.
\end{definition}

Note that when the data owners make their decisions on bidding, they have no idea about whether they would be selected or not. We assume that the data owners are rational and make decisions based on the expected benefit. 
Let $P_i=p_i q_i$ be the expected payment of $s_i$. The expected utility of a data owner $s_i$ is the difference between her expected compensation and the valuation of expected privacy loss, i.e., $u_i \coloneqq P_i - \varepsilon_i \theta_i q_i$.

The allocation function together with the payment function is expected to satisfy the following properties.

\begin{itemize}
\item {\noindent \bf Incentive compatibility (IC).}
Each data owner $s_i, 1\leq i \leq n$ maximises her expected utility when reporting true valuation, i.e., 
%Data owner $i$ with valuation $\theta_i$ is incentivised to report her true valuation, for $1\leq i \leq n$:
\begin{equation}
   u_i(\theta_i,\vec{b}_{-i}) \geq u_i(b_i,\vec{b}_{-i}) \text{  } \forall b_i \neq \theta_i, \forall \vec{b}_{-i}\in [\underline{\theta}, \overline{\theta}]^{n-1}. %\in [\underline{\theta},\overline{\theta}]
   \label{eqn:IC}
\end{equation}
\item {\noindent \bf Individual rationality (IR).} 
Each data owner $s_i, 1\leq i \leq n$ gets  non-negative expected utility when reporting true valuation, i.e., 
%Data owner $i$ is willing to participate in the system, and will never receive negative utility if she reports her valuation truthfully, for $1\leq i \leq n$:
\begin{equation}
    u_i(\theta_i,\vec{b}_{-i}) \geq 0 \text{  } \forall \vec{b}_{-i} \in [\underline{\theta}, \overline{\theta}]^{n-1}.
    \label{eqn:IR}
\end{equation}
\item {\noindent \bf Budget feasibility (BF).} The total expected payment is within the budget $\beta$ specified by the data consumer, i.e., 
\begin{equation}
    \sum_{i=1}^n p_i(\vec{b}) q_i(\vec{b}) \leq \beta.
    \label{eqn:BF}
\end{equation}
\end{itemize}

\noindent Intuitively, IC ensures the data owners are incentivised to report their true valuations as doing so leads to the best utility while IR ensures that the data owners are willing to participate in the system as it leads to at least non-negative utility. 

%{\red (Please note that I changed the definition of $p_i$ from expected payment to realised payment.)}

\begin{definition}
    A {\em query} $f: \mathbb{R}^n \to \mathbb{R}$ is a mapping from private data of $n$ data owners to a real value.
\end{definition}

Commonly used queries include count and median queries. Count query counts the number of $1$s in a binary dataset, i.e., $f(\vec{t}) = \sum^n_{i=1} t_i$, while median query returns the median number in a real-valued dataset. %, while linear predictor query computes the weighted sum of the data, given a public weight vector $\vec{w} = (w_1, \ldots, w_n)$, i.e., $\sum^n_{i=1} w_i t_i$. 

%{\red (a vector should be represented by $()$ rather than $\{\}$.)} 

Considering the privacy concern, we use local differential privacy (LDP)~\cite{duchi2013local} as the privacy concept. In order to ensure LDP, a {\em local randomiser} is commonly utilised, allowing data owners to obfuscate their private data before sharing it for answering queries.

Let $\mathcal{D}$ denote the domain of private data and $\Omega$ be a probability space. 
We define the local randomiser and LDP as follows.

%An LDP query process is usually equipped with a {\em local randomiser} where each data owner can obfuscate her private data before sharing it for answering queries.
%We take local differential privacy \cite{duchi2013local} as the privacy concept and the measure of privacy loss of data owners. 
%A local randomiser takes data owners' real private data as input and generates an output that is perturbed with noise. 

\begin{definition}
%(Local Randomiser) 
A {\em local randomiser} $\LR: \mathcal{D} \times \Omega \to \mathcal{D}$ is a randomised function mapping a private value $t_i$ to a random value $t_i'\in \mathcal{D}$.
\end{definition}

\begin{definition}[%($\varepsilon_i$-local differential privacy ($\epsilon_i$-LDP) 
\cite{duchi2013local}] A local randomiser $\LR$ is {\em $\varepsilon_i$-local differentially private}, if for any pair of input $t_i, t_i' \in \mathcal{D}$ and for any possible output $o \in Range(\LR)$, we have 
\[
\Pr[\LR(t_i) = o] \leq e^{\varepsilon_i} \Pr[\LR(t_i') = o],
\]
\end{definition}
\noindent where $\varepsilon_i$ is a non-negative real number to measure data owner $s_i$'s privacy loss.

% In a PDQS $M$, the procurement process, due to the budget, may not afford to buy the data from all data owners, while the query process adds noisy to data. Both lead to inaccuracy of the query answer.
%{\teal The query process adds noise to data, which sacrifices the accuracy of the query answer. Here, } 
For a given query $f$ and PDQS $M$, %{\red (What is $M$?)}, 
we use $(\alpha,\delta)$-probably approximately correct  ($(\alpha,\delta)$-PAC) to measure $M$'s accuracy. 

%{\red (M: Do we have any results on PAC accuracy? It seems like we do not use this definition below.) I use PAC accuracy when we mention PPEM.}

\begin{definition}
For $\alpha,\delta\in [0,1]$, a private data query system $M$ is {\em $(\alpha,\delta)$-probably approximately correct}  ($(\alpha,\delta)$-PAC) if for any dataset $\vec{t} = (t_i, \ldots, t_n)$,
$$\Pr[|M(\vec{t})-f(\vec{t})|\geq \alpha] \leq 1-\delta,$$
\end{definition}
\noindent where $|M(\vec{t})-f(\vec{t})|$ is the difference between the query answer derived by $M$ and the true answer.

\bigskip

The goal of this research is to develop a PDQS that meets IC, IR, BF, and $\varepsilon_i$-LDP properties, while also approximating query accuracy. %, without utilising a data broker.

\section{Proposed framework}
\label{sec:framework}
%We design a framework of private data query systems (PDQS) that incentivises data owners to participate in the system (IR) and truthfully report their valuations (IC). This system does not rely on a data broker, and provides privacy protection for data owners' private data ($\varepsilon_i$-LDP), as shown in Algorithm \ref{alg:framework}. In addition, proper instantiation of this framework may approximate the query accuracy given a fixed budget (BF).

%{\red (Describe the intuitive idea of our framework. How do we embed the procurement process and query process to address the issues we claimed in Introduction.) } 

%The proposed framework takes the bids, private data, budget, query and a decreasing allocation function as inputs, where the \textit{decreasing allocation function} is defined as the following. 

%\begin{definition}(Decreasing allocation function)
%A decreasing allocation function $q$ is an allocation function whose output $(q_1, \ldots, q_n)$ satisfies that $q_i \in [0, 1)$ for any $i \in [n]$, and $q_i \geq q_j$, for any data owner $i, j$ and $b_i \leq b_j$.
%\end{definition}

%A decreasing allocation function decreases on any input bid, and always outputs probabilities that are less than 1.

%{\red (Why doesn't a non-increasing function work? Why not define the output of the allocation function as $[0,1)$ in the problem formulation part? Why does the range have to be less than 1?)}

We propose the framework \textit{Integrated PDQS}, which combines the procurement and query processes into an integrated system. Integrated PDQS receives bids, private data, the budget, the query and an allocation function as inputs, and provides a payment result and a query answer as outputs. This framework can be instantiated to create PDQSs that satisfy the desired properties, i.e., IC, IR, BF and $\varepsilon_i$-LDP. The workflow of the framework is shown in Figure~\ref{fig:framework}. 

Unlike existing PDQSs that separate the procurement and query processes, the Integrated PDQS assigns probabilities to each data owner. These probabilities determine both the likelihood of being selected by the system and the likelihood of reporting their true private data. %This approach offers the benefit of compensating only those data owners who report their true private data. 
Consequently, the consumer's budget is utilised more efficiently, leading to higher query accuracy.\footnote{Note that the randomised approach, under the premise of privacy protection, achieves the goal of compensating only those owners who report true data, which deterministic procurement mechanisms cannot accomplish.}

%We present the framework of PDQS that implements the procurement and query processes. The PDQS takes bids, private data, the budget, the query, and an allocation function as inputs, then returns a payment result and a query answer. The design is motivated by the issue of the separation between the procurement and the query processes, as discussed in Section~\ref{sec:intro}. We propose a new PDQS framework that embeds the procurement process into the query process. See the workflow of the proposed framework in Figure~\ref{fig:framework}. 

%{\red (M: Here I removed the allocation vector as the output as it is confused about the probability and the final allocation result. )}

\begin{figure}
    \centering
    \includegraphics[width=0.8\columnwidth]{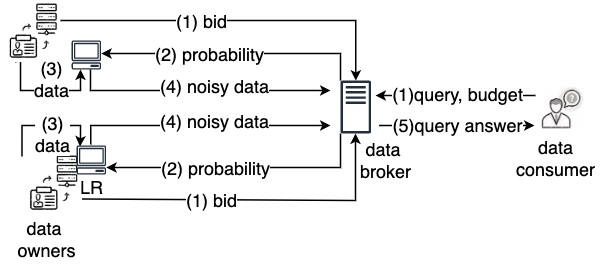}
    \caption{Workflow of proposed Integrated PDQS framework}
    \label{fig:framework}
\end{figure}

Algorithm~\ref{alg:framework} outlines the specific steps of the Integrated PDQS framework. Initially, the framework calculates the allocation probabilities and expected payments for each data owner (Lines 1-2 of Alg.~\ref{alg:framework}). Unlike existing approaches that select data owners before the query process, the Integrated PDQS framework incorporates an {\em Integrated Local Randomiser} (ILR) $\mathcal{L}_I$ for each data owner  (Lines 3-5 of Alg.~\ref{alg:framework}). As shown in Alg.~\ref{alg:LR}, $\LR_I$ is employed to select a data owner based on her probability $q_i$. If a data owner is chosen, her true private data is reported, and her compensation is calculated as $P_i/q_i$. On the other hand, if a data owner is not selected, a random value is reported, and no compensation is provided. Subsequently, the framework performs the query on the noisy dataset obtained from the data owners (Line 6 of Alg.~\ref{alg:framework}).

In the framework, the procurement process is embedded in the query process through the use of ILR. In contrast to the traditional local randomiser, ILR not only determines the reported private data but also determines the compensation for each data owner. This design ensures that only the data owners who disclose their true data are selected and receive payment.

\begin{algorithm}
\begin{algorithmic}[1]
\caption{The Integrated PDQS framework}
\label{alg:framework} \footnotesize
\Require Data owners $\vec{s}$, budget $\beta$, query $f$, allocation function $q$
\Ensure %Allocation vector $(q_1, \ldots, q_n)$, 
Payment vector $(p_1, \ldots, p_n)$, query answer $z$
\State Generate the allocation vector $(q_1, \ldots, q_n)$  %= q(\vec{b}, \beta)$
\State Generate the payment vector $(P_1, \ldots, P_n)$%, where $p_i = b_i q_i \ln(\frac{1+q_i}{1-q_i}) + \int^{\overline{b}_i}_{b_i}q(x, b_{-i})\ln(\frac{1+q_i}{1-q_i}) dx$
%\STATE Select a subset of data owners $\vec{s}_c$ according to the allocation vector, payment vector and budget
\For{each data owner $s_i$}
\State Get a random value and the payment $t'_i, p_i = \mathcal{L}_I(t_i, q_i, P_i)$
\EndFor
\State Compute the query answer $z = f(t'_1, \ldots, t'_n)$
\State Return $(p_1, \ldots, p_n), z$ %Return $(q_1, \ldots, q_n), (p_1, \ldots, p_n), z$
\end{algorithmic}
\end{algorithm}

\begin{algorithm}
\caption{Ingerated Local Randomiser $\mathcal{L}_{I}$}
\label{alg:LR} \footnotesize
\begin{algorithmic}[1]
\Require Private data $t_i$, probability $q_i$, expected payment $P_i$
\Ensure Perturbed data entry $\hat{t}_i$ 
%\State {\red Set $\hat{t}_i=
% \begin{cases}
%			t_{i} & \text{ with probability } q_i \\
%            1-t_{i} & \text{ with probability } 1-q_i
%		 \end{cases}$}
\State With probability $q_i$, set  $\hat{t}_i = t_i,  p_i = P_i / q_i$; with probability $1-q_i$, set $\hat{t}_i$ to be a random value $t_i' \in \mathcal{D}, p_i = 0$
\State Return $\hat{t}_i, p_i$
\end{algorithmic}
\end{algorithm}

%We now demonstrate how the proposed framework satisfies the desired properties. Firstly, we show that the local randomiser $\mathcal{L}$ utilised in the framework satisfies $\varepsilon_i$-LDP, regardless of the choice of the allocation function. 

Applying the ILR with probability $q_i$ for data owner $s_i$ ensures $\ln{\frac{1+q_i}{1-q_i}}$-local differential privacy, as shown in the following lemma.

%{\red (M: I have changed the line 1 of algorithm 2. Please check whether the lemma is correct. Also, revise it if the domain is changed. ) }

\begin{lemma}
\label{lemma:lr}
    The Integrated Local Randomiser $\mathcal{L}_I$  with probability $q_i$ is $\varepsilon_i$-local differential privacy, where $\varepsilon_i= \ln{\frac{1+q_i}{1-q_i}}$. 
\end{lemma}

\begin{proof}
Given two data entries $t_i,t_i'$, and $t_i \neq t_i'$, %let $\varepsilon_i = \ln{(\frac{1+q_i}{1-q_i})}$, 
for $r\in \{0,1\}$: 

\noindent when $r=t_i$, 
\[
\frac{\Pr[\mathcal{L}_I(t_i)=r]}{\Pr[\mathcal{L}_I(t_i')=r]}=\frac{q_i + \frac{1}{2}(1-q_i)}{\frac{1}{2} (1-q_i)} = \frac{1+q_i}{1-q_i}\leq e^{\varepsilon_i};
\]

\noindent when $r=t_i'$, 
\[
\frac{\Pr[\mathcal{L}_I(t_i)=r]}{\Pr[\mathcal{L}_I(t_i')=r]}=\frac{\frac{1}{2} (1-q_i)}{q_i + \frac{1}{2}(1-q_i)} = \frac{1-q_i}{1+q_i}  \leq e^{\varepsilon_i}.
\]
\end{proof}

Note that when $q_i$ approaches $1$, the privacy guarantee $\varepsilon_i=\ln{\frac{1+q_i}{1-q_i}}$ is meaningless as it is infinity. Hence, from now on, we require the allocation probability $q_i\in [0,1)$.

\medskip

%\begin{align}
%&\max &\sum^n_{i=1} q_i \ln{\frac{1+q_i}{1-q_i}}    \\
%& \text{s.t. } & \text{IC, IR, and BF} 
%\end{align}

Our goal is to design PDQSs that implement the Integrated PDQS framework, and satisfy IC, IR, BF and $\varepsilon_i$-LDP while approximating query accuracy. Any PDQS that instantiates the proposed framework employs ILR, so $\varepsilon_i$-LDP is guaranteed. 

For query accuracy, we introduced the accuracy notion of $(\alpha, \delta)$-PAC in Section~\ref{sec:problem}. To achieve high accuracy, we adopt the principle of Purchased Privacy Expectation Maximisation (PPEM) proposed in \cite{zhang2020selling}. The PPEM principle highlights that acquiring a certain amount of expected privacy is a necessary condition for attaining a particular level of PAC accuracy. Thus, we may maximise the total expected purchased privacy to approximate the query accuracy.
According to Lemma~\ref{lemma:lr}, the PDQS that adopts the proposed framework purchases the privacy of  $\ln{\frac{1+q_i}{1-q_i}}$ with a probability $q_i$ from data owner $s_i$. As a result, the total expected purchased privacy in the PDQS is given by $\sum^n_{i=1} q_i \ln{\frac{1+q_i}{1-q_i}}$.

Then, we characterise IC and IR properties. We define $w_i(b_i, b_{-i}) \coloneqq q_i(b_i, b_{-i}) \ln{\frac{1+q_i(b_i, b_{-i})}{1-q_i(b_i, b_{-i})}}$ as the expected privacy loss of data owner $s_i$, when she bids $b_i$ and others bid $b_{-i}$. For brevity, we refer to $w_i(b_i, b_{-i})$ as $w_i$ when the context is clear.
Applying Archer and Tardos's theorem~\cite{archer2001truthful} and considering $w_i$ as the amount of \textit{load} assigned to each data owner, we can derive the following theorem:

%{\teal (M: why did you make the revisions here? Archer and Tardos did not mention $\ln{\frac{1+q_i(x, b_{-i})}{1-q_i(x, b_{-i})}}$.   )}

\begin{theorem}[\cite{archer2001truthful}]
\label{theorem:Archers}
A PDQS $M$ satisfies IC and IR if and only if 
\begin{enumerate}
    \item the expected privacy loss $w_i(b_i, b_{-i})$ does not increase with respect to the bid $b_i$,
    \item the expected privacy loss $w_i(b_i, b_{-i})$ satisfies $\int^{\overline{\theta}}_0 w_i(x, b_{-i}) dx < \infty$ for all $i, b_{-i}$, and
    \item the expected payment is in the form of $P_i = b_i w_i(b_i, b_{-i}) + \int^{\overline{\theta}}_{b_i} w_i(x, b_{-i}) dx$.  
\end{enumerate}
\end{theorem}
%Archer and Tardos \cite{archer2001truthful} demonstrated that a mechanism satisfies IC and IR if and only if (1) the allocation $q_i(b_i, b_{-i})$ decreases with respect to the bid $b_i$ and (2) satisfies \textcolor{cyan}{$\int^{\overline{\theta}}_0 q_i(x, b_{-i}) \ln{\frac{1+q_i(x, b_{-i})}{1-q_i(x, b_{-i})}} dx < \infty$ for all $i, b_{-i}$, while (3) the expected payment is in the form of $P_i = b_i q_i(b_i, b_{-i})\ln{\frac{1+q_i(b_i, b_{-i})}{1-q_i(b_i, b_{-i})}} + \int^{\overline{\theta}}_{b_i} q_i(x, b_{-i}) \ln{\frac{1+q_i(x, b_{-i})}{1-q_i(x, b_{-i})}} dx$}. 

Therefore we construct the optimisation problem that maximises the total expected purchased privacy while satisfying the IC, IR and BF constraints.

As the privacy loss $\varepsilon_i = \ln{\frac{1+q_i}{1-q_i}}$ of data owner $s_i$ when she is selected with probability $q_i$ increases with respect to $q_i$ in the domain of $[0,1)$, we can observe that Condition 1 of Theorem~\ref{theorem:Archers} is satisfied if the allocation function $q$ is non-increasing. Therefore, in the subsequent discussion, we impose the requirement of a non-increasing allocation function to fulfil Condition 1.

%replace the IC and IR constraints and plug the payment equation into the BF constraint in Problem~\ref{eqn:opt1}, and we get 

\begin{equation}\label{eqn:opt2}
\begin{aligned}
    &\max & &\sum^n_{i=1} q_i \ln{\frac{1+q_i}{1-q_i}}  \\
    &\text{s.t. } & &q_i(b_i', b_{-i}) \leq q_i(b_i, b_{-i}),  \forall  b_i' > b_i \\
    %& & &\sum_{i=1}^n [b_i q_i(b_i, b_{-i}) + \int^{\overline{\theta}}_{b_i}q_i(x, b_{-i})dx] \leq \beta \\
    & & & \sum_{i=1}^n [b_i w_i + \int^{\overline{\theta}}_{b_i} w_i(x, b_{-i}) dx] \leq \beta \\
    & & &0 \leq q_i < 1,  \forall i \in [0, n] 
\end{aligned}
\end{equation}

%\noindent Equation (\ref{eq:obj}) aims to maximise the total amount of expected privacy purchased by the mechanism. The constraint (\ref{eq:oconstraint1}) represents the BF property, requiring the total expected payment to be no more than the given budget. The constraint (\ref{eq:oconstraint2}) and (\ref{eq:oconstraint3}) require the allocation function is decreasing and the output probabilities to be in the range $[0,1)$, to satisfy IC and IR constraint. We do not have any constraint for LDP, because our proposed framework ensures $\varepsilon_i$-LDP by construction.

Now the problem is to find the proper allocation function $q$ to solve \eqref{eqn:opt2}. In the following sections, we propose two  instances of the proposed framework to solve the problem.

\section{Greedy Private Query Mechanism (GPQM)}\label{sec: GPQM}
%In order to address the optimisation problem outlined earlier, we implement the proposed framework and introduce the Greedy Private Query Mechanism (GPQM), which is illustrated in \ref{alg:GPQM}.

The core of designing an instance to implement the Integrated PDQS framework is the design of an allocation function $q$ that addresses \ref{eqn:opt2}. A straightforward idea is to deploy a non-increasing function as the allocation function such that a data owner with a high valuation has a low chance to be selected, and then, according to the distribution, greedily choose data owners until the budget is used up. Such non-increasing allocation function can be a linear function e.g. $q_i=1-b_i$, a logarithmic function e.g. $q_i = -\log(b_i)$, or an exponential function e.g. $q_i = e^{-b_i}$. We refer to the greedy instance as {\em {\bf g}reedy {\bf p}rivate {\bf q}uery {\bf m}echanism} (GPQM).

To be specific, given a non-increasing allocation function, GPQM first determines the allocation probability $q_i$. Also, for each data owner $s_i$ whose expected privacy loss is $w_i = q_i \ln{\frac{1+q_i}{1-q_i}}$ %privacy loss is $\varepsilon_i = \ln\frac{1+q_i}{1-q_i}$
, set her expected payment as $P_i = b_i w_i + \int^{\overline{\theta}}_{b_i}w_i(x, b_{-i})dx$.
%\textcolor{cyan}{$P_i = b_i q_i(b_i) \varepsilon_i + \int^{\overline{\theta}}_{b_i}q_i(x)\ln{\frac{1+q_i(x)}{1-q_i(x)}} dx$}. 
Then GPQM sorts the data owners by their allocation probabilities in descending order. Following the order, each data owner runs the ILR $\LR_I$ until the budget is used up. GPQM finally returns the payment vector and the query answer. See Algorithm~\ref{alg:GPQM}.

%{\red (todo: check the definitions of allocation, payment, allocation probability, expected payment. )}

%\textcolor{cyan}{Notes: In the definition of BF constraint, we use $p_i$ to denote the real payment and $p_iq_i$ to denote the expected payment. I modified Alg~\ref{alg:GPQM} to make it consistent. We need to discuss this part.}

\begin{algorithm}
\begin{algorithmic}[1]
\caption{Greedy Private Query Mechanism (GPQM)} \footnotesize
\label{alg:GPQM}
\Require Data owners $\vec{s}$, budget $\beta$, query $f$, non-increasing allocation function $q$
\Ensure %Allocation vector $(q_1, \ldots, q_n)$, 
Payment vector $(p_1, \ldots, p_n)$, query answer $z$
\State Generate the allocation vector $(q_1, \ldots, q_n) = q(\vec{b})$ 
\State Compute the expected payment vector $(P_1, \ldots, P_n)$, where $P_i = b_i w_i + \int^{\overline{\theta}}_{b_i}w_i(x, b_{-i})dx$
%\State Compute the expected payment vector $(p_1q_1, \ldots, p_nq_n)$, where $p_iq_i = \left(b_i q_i(b_i)  + \int^{\overline{\theta}}_{b_i}q_i(x) dx \right) \varepsilon_i$
\State Sort the data owners in descending order with respect to $q_i$
\State Initialise $k=1$
\While{$\sum_{i=1}^k P_i \leq \beta$}
\State Get a random value and the payment $t'_i, p_i = \mathcal{L}_I(t_i, q_i, P_i)$
\State Increment $k=k+1$
\EndWhile
\State Set $p_i = 0$, if $i > k$
%\State Set $p_i=p_i/q_i, i\leq k$; $p_i=0, i> k$}
\State Compute the query answer $z = f(t'_1, \ldots, t'_n)$
\State Return $(p_1, \ldots, p_n), z$
\end{algorithmic}
\end{algorithm}
%
%GPQM receives bids, private data, the budget, query, and a decreasing allocation function as its inputs. 
%
%The procurement process is described in Lines 1 to 7. Firstly, the allocation and payment vector is generated following the proposed framework. Then, the data owners are sorted in descending order based on the value of $q_i$, which ensures that those with a lower unit price and a higher likelihood of being selected are ranked higher. The mechanism then determines the last data owner, denoted as $j$, such that the budget is insufficient to select any data owner after $j$. For those data owners after $j$, their probabilities of being selected and expected payments are set to 0. Lines 3 to 7 can be considered as complementary to Line 1, and together they constitute the allocation function.
%
%The query process is illustrated in Lines 8 to 11. A local randomiser described in Algorithm \ref{alg:LR} is applied to each data owner, following the proposed framework. Then the query is issued on the collected dataset. The system finally returns the allocation and payment vectors and the query answer.

%{\red M: Please check the proof as the proofs in the previous section has been removed. A minor problem: we didn't define $\overline{b}_i$. Change it to $\overline{\theta}_i$. }

\begin{lemma}
    Greedy private query mechanism (GPQM) is IC and IR.
    \label{lem:GPQM_IC_IR}
\end{lemma}

\begin{proof}
We prove the lemma by Archer and Tardos' theorem \cite{archer2001truthful}. As GPQM employs a non-increasing allocation function, i.e., the allocation $q_i$ does not increase with respect to the bid $b_i$, and the expected payment $P_i$ is in the form specified in Line 2 of Alg.~\ref{alg:GPQM}, Conditions 1 and 3 are met. 
%As GPQM employs a non-increasing allocation function, the allocation $q_i$ does not increase with respect to the bid $b_i$. Also, the expected payment $P_i$ is in the form specified in Line 2 of Alg.~\ref{alg:GPQM}. According to Archer and Tardos' theorem \cite{archer2001truthful}, GPQM is IC and IR if $\int^{\overline{\theta}}_0 q_i(x, b_{-i}) \ln{\frac{1+q_i(x, b_{-i})}{1-q_i(x, b_{-i})}} dx < \infty$ for all $i, b_{-i}$.

Now we show that Condition 2 is also met. 
We rewrite the integral as $\int^{\overline{\theta}}_0 q
_i(x)\ln{(1+q_i(x))} - q_i(x)\ln{(1-q_i(x))}dx$, where $q_i(x)$ is short of $q_i(x, b_{-i})$. %{\teal (The abb. has been used before. We need to define it earlier.) }
 The former part is finite, so we focus on the latter part. We have 
\begin{align*}
& \int^{\overline{\theta}}_0 -q(x)\ln{(1-q(x))}dx \leq \int^{\overline{\theta}}_0 -\ln{(1-q(x))}dx \\
&\leq -x\ln{(1-\overline{q})}
\leq -\overline{\theta}\ln{(1-\overline{q})},   
\end{align*}
where $\overline{q}$ is the maximum value of $q(x)$ in its domain. As the range of $q$ is $[0,1)$, we have $\overline{q} = 1- \tau$, where $\tau$ is a very small positive real. We then have $-\ln{(1-\overline{q})} \leq -\ln \tau$, which is finite.

    %We first show that the allocation $q_i$ does not increase with respect to the bid $b_i$ for any $s_i \in \vec{s}$. For bidder $s_i$ such that $i \leq k$, the allocation $q_i$ is decided by a non-increasing allocation function $q$ so that $q_i$ does not increase with $b_i$. For bidders $s_i$ such that $i > k$, the allocation $q_i$ is set to be 0 so that }

\end{proof}

\begin{theorem}
The greedy private query mechanism (GPQM) is IC, IR, BF and LDP.
\end{theorem}

The mechanism is BF by construction. As the ILR $\mathcal{L}_I$ is applied to perturb the private data, the mechanism satisfies LDP.

%GPQM is intuitive and light, whereas a proper allocation function needs to be designed before applying the mechanism, which can be critical. 

\section{The neural-based private query mechanism (NPQM)}
\label{sec: NPQM}

% GPQM ensures BF and LDP through its construction and satisfies IC and IR if the allocation function is monotonic. However, selecting an appropriate allocation function from a range of eligible functions and determining the function's parameters are both critical factors that impact the mechanism's accuracy, and optimal choices may vary in different scenarios. 
We introduce another instance that
%We introduce the neural-based private query mechanism (NPQM), which 
implements the proposed framework. %NPQM satisfies IC, IR, BF and $\varepsilon_i$-LDP constraints, and provides an approximation of the query accuracy under a fixed budget. 
Unlike GPQM, this instance addresses \eqref{eqn:opt2} by learning the allocation function that approximates the total purchased privacy while satisfying the given constraints. To achieve this, we apply QMIX~\cite{rashid2020monotonic} and neural-network techniques to design and parameterise the allocation function, and learn its parameters using dual ascent algorithm~\cite{boyd2011distributed}. %{\red (M: do we need a citation here?)}. 
We refer to this instance as {\em {\bf n}eural-based {\bf p}rivate {\bf q}uery {\bf m}echanism} (NPQM).

\subsection{Allocation function design and parametrisation}
% how we design the allocation function

%The allocation function $q(\vec{b}) \colon [\overline{\theta}, \underline{\theta}]^n \to [0, 1]^n$ is a mapping from data owners' bids to the probabilities of selecting each data owner. 

The allocation function in NPQM is designed as $q_i(\vec{b}) = \sigma(|w_2|(-|w_1|b_i+c)+d)$, where $\sigma$ denotes the sigmoid function, i.e., $\sigma(x) = \frac{e^x}{e^x+1}$, and $w_1, w_2, c, d$ are parameters. 
%{\red (M: These two sentences and the fact below are explanations for how the IC, IR properties hold. Consider put these into the proof. Here, we just need to introduce the design of the function, such as the value of $w_1, w_2$. )}

% how we represent the parameters
The parameters $w_1, w_2, c, d$ are generated by separate neural networks. Specifically, we use three separate three-layer neural networks to generate $w_1$, $w_2$, and $d$, respectively. The neural network is composed of three functions, which are:

\begin{itemize}
    \item A linear function $l^{(1)}\colon [\overline{\theta}, \underline{\theta}]^n \rightarrow \mathbb{R}^{h}$ takes the bids as input. It has the form $l^{(1)}(\vec{b})=A^{(1)}\vec{b} + k^{(1)}$ where $A^{(1)}$ and $k^{(1)}$ are coefficients of the linear function, and $h$ is a constant denoting the number of neural employed in Layer 2;
    \item A function $l^{(2)}\colon \mathbb{R}^{h} \rightarrow \mathbb{R}^{h}$ takes the result from Layer 1. In detail, $l^{(2)}(x)=x$ if $x \geq 0$; $l^{(2)}(x)=0$ otherwise. This is known as a ReLU function;
    \item A linear function $l^{(3)}\colon \mathbb{R}^h \rightarrow \mathbb{R}$ takes the the result from Layer 2. It has the form $l^{(3)}(\vec{b})=A^{(3)}\vec{b} + k^{(3)}$ where $A^{(3)}$ and $k^{(3)}$ are coefficients of the linear function.
\end{itemize}

The parameters $w_1,w_2,d$ are computed by $l^{(3)}(l^{(2)}(l^{(1)}(\vec{b})))$, which forms a function from $[\overline{\theta}, \underline{\theta}]^n$ to $\mathbb{R}$. Notice that although three parameters $w_1,w_2$, and $d$ are calculated by the same form of function, they have different values since the coefficients in $l^{(3)}$ and $l^{(1)}$ are different, e.g. $A^{(1)}$ for $w_1$ and $w_2$ can be different.

The parameter $c$ is directly computed by a single linear function $l^{(c)}: [\overline{\theta}, \underline{\theta}]^n \rightarrow \mathbb{R}$ with the form $l^{(c)}(\vec{b})=A^{(c)}\vec{b} + k^{(c)}$. Let $\mu = (w_1, w_2, c, d)$ denote all learnable parameters in the parameterised model, then we denote the parameterised allocation function as $q^\mu$. 
% and use $N^\mu := (q^\mu_1,\cdots,q^\mu_n)$ to denote our parameterised model in the later sections. WD: We may just use $q^\mu$ instead of $N^\mu$.

This type of function and neural network design has been previously utilised in QMIX \cite{rashid2020monotonic}, a well-known deep multi-agent reinforcement learning algorithm, which ensures the monotonicity between the global and the local value function. The details of QMIX are described in Appendix~\ref{appendix: QMIX}.

The optimal values of the coefficients in each layer will be trained and approximated by the Stochastic Gradient Descent (SGD) method. The performance of the neural network thus can be guaranteed by the proper settings of coefficients in each layer.

\subsection{Learning the parameters of the allocation function}

We employ dual ascent techniques~\cite{boyd2011distributed} to learn the parameter $\mu$ specified above and approximate the optimal solution to \eqref{eqn:opt2}. We begin with rewriting \eqref{eqn:opt2} with $\mu$ as the following.

%Problem~\eqref{eqn:opt2} can be rewritten as the following

%\resizebox{0.97\columnwidth}{!}{
\begin{align} \label{pro: primal}
    \max_{\mu}  &\sum^{n}_{i=1} q_i^{\mu}\ln{\frac{1 + q_i^{\mu}}{1 - q_i^{\mu}}} \\
    \text{s.t. }  & q_i^{\mu}(b_i', b_{-i}) \leq q_i^{\mu}(b_i, b_{-i}), \forall b_i'> b_i \label{eq:constraint2} \\
      & \sum^n_{i=1} \left(b_i q_i^{\mu}\ln{\frac{1 + q_i^{\mu}}{1 - q_i^{\mu}}} + \int^{\overline{\theta}}_{b_i} q^{\mu}(x) \ln{\frac{1 + q^{\mu}(x)}{1 - q^{\mu}(x)}}dx\right) - \beta \leq 0 \label{eq:cons_para}  \\
      & 0 \leq q_i^{\mu} < 1, \forall i \in [0, n]    \label{eq:constraint3}
\end{align}
%}

%\resizebox{0.97\columnwidth}{!}{
%\begin{equation}
%\begin{aligned}
%    &\max_{\mu} & &\sum^{n}_{i=1} q_i^{\mu}\ln{\frac{1 + q_i^{\mu}}{1 - q_i^{\mu}}} \\
%    &\text{s.t.} & & q_i^{\mu}(b_i', b_{-i}) - q_i^{\mu}(b_i, b_{-i}) < 0, \forall b_i'> b_i, \\
%    & & & \sum^n_{i=1} \left(b_i q_i^{\mu}\ln{\frac{1 + q_i^{\mu}}{1 - q_i^{\mu}}} + \int^{\overline{\theta}_i}_{b_i} q^{\mu}(x) \ln{\frac{1 + q^{\mu}(x)}{1 - q^{\mu}(x)}}\right) - \beta \leq 0  \\
%    & & & 0 \leq q_i^{\mu} < 1, \forall i \in [0, n]  
%\end{aligned}
%\end{equation}
%}

%\begin{equation}
%\label{pro: primal}
%    \max_{\mu} \sum^{n}_{i=1} q_i^{\mu}\ln{\frac{1 + q_i^{\mu}}{1 - q_i^{\mu}}},
%\end{equation}
%
%\begin{align}
%\label{eq:cons_para}
%    \text{s.t.}
%    \sum^n_{i=1} \left(b_i q_i^{\mu}\ln{\frac{1 + q_i^{\mu}}{1 - q_i^{\mu}}} + \int^{\overline{b}_i}_{b_i} q^{\mu}(x) \ln{\frac{1 + q^{\mu}(x)}{1 - q^{\mu}(x)}}\right) - \beta \leq 0,
%\end{align}
%
%\begin{equation}\label{eq:constraint2}
%    \forall i > j, q_i^{\mu}(b_j, b_{-i}) - q_i^{\mu}(b_i, b_{-i}) < 0,
%\end{equation}
%\begin{equation}\label{eq:constraint3}
%    \forall i \in [0, n], 0 \leq q_i^{\mu} < 1
%\end{equation}
%
%{\red (!!!! M: Please make sure the notations are consistent. For instance, data owner is represented by $s_i$ rather than $i$.)}
Recall that $n$ is the number of data owners, and $q_i$ is the probability that data owner $s_i$ is selected by the system. %when bidding $b_i$, given other data owners bidding $b_{-i}$. 
The objective function aims to maximise the total expected privacy purchased by the system, where $\ln{\frac{1 + q_i^{\mu}}{1 - q_i^{\mu}}}$ is the privacy loss $\varepsilon_i$ of data owner $s_i$, if she is selected %by the system. 
Constraint (\ref{eq:constraint2}) ensures that, for any data owner $s_i$, the allocation function is non-increasing with respect to her bid $b_i$, given other data owners' bids are fixed.
Constraint (\ref{eq:cons_para}) %{\color{magenta}Libo 1-5: Is this constraint 13? why not constraint 9?} 
ensures that the total expected payment does not exceed the given budget. Constraint (\ref{eq:constraint3}) ensures that the probability of a data owner being selected is less than 1, thereby providing privacy protection for each data owner.

Constraint (\ref{eq:constraint2}) is satisfied since the parameterised model is enforced to be a monotonic function.
Constraint (\ref{eq:constraint3}) is also satisfied as we used the Sigmoid function, and the output is constrained between $0$ and $1$. Thus, we omit these two constraints in the following formulations. 

Next, we discuss how to update the parameters of the parameterised allocation function $q^\mu$ to satisfy the BF constraint (as shown in  (\ref{eq:cons_para})) %{\color{magenta}Libo 1-5: Why constraint 13?}
and approximate the total purchased privacy (as shown in (\ref{pro: primal})) by the dual ascent algorithm~\cite{boyd2011distributed}. 
%{\red (M: do we need a citation here?)} WD: cited

We establish the dual problem of the primal problem (\ref{pro: primal}). Let $\phi(\mu) = \sum^n_{i=1} q_i^{\mu}\ln{\frac{1 + q_i^{\mu}}{1 - q_i^{\mu}}}$ denote the objective function of the primal problem, representing the total expected purchased privacy of the system. Let $g(\mu)=\sum^n_{i=1} \left(b_i q_i^{\mu}\ln{\frac{1 + q_i^{\mu}}{1 - q_i^{\mu}}} + \int^{\overline{\theta}}_{b_i} q^{\mu}(x) \ln{\frac{1 + q^{\mu}(x)}{1 - q^{\mu}(x)}}dx\right) - \beta$ denote the difference between the total expected payment and the budget.
%{\color{magenta}Libo: Could we say clearly what is $g(\mu)$?}. WD: I specified the expression.
The standard form of the primal problem is

\vspace{-0.5cm}

\begin{align}
    \min_{\mu} & -\phi(\mu) \label{eq:minmu}\\ 
    \text{s.t. } & g(\mu) \leq 0 \label{eq:cons}
\end{align}

%\begin{equation}
%    \min_{\mu} -\phi(\mu),
%\end{equation}
%\begin{align}
%\label{eq:cons}
%    \text{s.t. }
%    g(\mu) \leq 0.
%\end{align}

The Lagrangian is $L(\mu, \lambda) = -\phi(\mu) + \lambda g(\mu)$
%\begin{equation}
%\begin{aligned}
%    & \mathcal{L}(\mu, \lambda) = - \sum^n_i \Phi_i^{\mu}\ln{\frac{1 + \Phi_i^{\mu}}{1 - \Phi_i^{\mu}}} \\
%    & +\lambda[\sum^n_{i=1} (b_i \Phi_i^{\mu}\ln{\frac{1 + \Phi_i^{\mu}}{1 - \Phi_i^{\mu}}} + \int^{\overline{b}_i}_{b_i} \Phi^{\mu}(x) \ln{\frac{1 + \Phi^{\mu}(x)}{1 - \Phi^{\mu}(x)}}) - \beta],
%\end{aligned}
%\end{equation}
% \begin{equation}
% \label{eq:lag}
%     L(\mu, \lambda) = -\phi(\mu) + \lambda g(\mu)
% \end{equation}
, where $\lambda \geq 0$ is the Lagrangian multiplier. The dual function of the optimisation problem defined by~\eqref{eq:minmu}, \eqref{eq:cons} is $\psi(\mu, \lambda) = \text{inf } L(\mu, \lambda)$ 
% \begin{equation*}
%     \begin{aligned}
%         \psi(\mu, \lambda) = \text{inf } L(\mu, \lambda)
%     \end{aligned}
% \end{equation*}
which denotes the infimum of the Lagrangian. Then we can build the dual problem: 

\begin{equation}\label{prob: dual}
\begin{aligned} 
    &\max & \psi(\mu, \lambda) \\ 
    &\text{ s.t. }  &\lambda \geq 0
\end{aligned}
\end{equation}
%\begin{equation}
%\label{prob: dual}
%    \text{maximise } \psi(\mu, \lambda),
%\end{equation}
%\begin{equation}
%    \text{s.t. } \lambda \geq 0
%\end{equation}
\noindent It aims to obtain the greatest lower bound on the solution to the primal problem. With the dual problem, solving the optimisation problem becomes easier.

We then apply dual ascent techniques to approach the dual problem, updating the Lagrangian multiplier $\lambda$ by maximising $\psi(\mu, \lambda)$ and updating $\mu$ by minimising $L(\mu, \lambda)$ interchangeably by their gradients, as the following:
\begin{equation*}
    \lambda = \lambda + \alpha \nabla_\lambda L(\mu, \lambda),
\end{equation*}
\begin{equation*}
    \mu = \arg \min_{\mu} L(\mu, \lambda),
\end{equation*}
where $\alpha$ is the learning rate.

\begin{algorithm}
\begin{algorithmic}[1]
\caption{Dual Ascent}
\label{alg:train} \footnotesize
\Require Training set $\vec{b}'$, budget $\beta$, learning rate $\alpha$, number of episodes $T$
\Ensure Allocation function $q^\mu$
%\STATE Sort data owners in ascending order with respect to $b_i$
\State Set $t=0$, and initialise $\lambda$, $q^\mu$. Set $BF = False$
\While{$t \leq T$}
\State Generate allocation vector $(q_1, \cdots, q_n) = q^\mu(\vec{b}')$
\State Compute Lagrangian $L(\lambda, \mu)$ 
\State Let $\mu' = \arg \min_{\mu} L(\lambda, \mu)$
\If{BF constraint (\ref{eq:cons_para}) is satisfied}
    \State Set $\mu=\mu'$, $BF=True$
\EndIf
\State $\lambda = \lambda + \alpha \nabla_\lambda L(\mu, \lambda)$, $t = t + 1$
\EndWhile
\State Return $q^\mu$ if $BF$, NO ANSWER otherwise

\end{algorithmic}
\end{algorithm}

To be specific, the dual ascent, as shown in  Algorithm \ref{alg:train}, takes the training set, budget, learning rate and the number of episodes as input, where the training set is generated from the same distribution and has the same size as the real bids of data owners. The algorithm initialises $t$, $\lambda$ and $q^\mu$ (Line 1 of Alg.~\ref{alg:train}). Then for each episode, the algorithm passes the training set to the allocation function $q^\mu$ and generates the allocation vectors used to compute the Lagrangian (Line 3 of Alg.~\ref{alg:train}). The algorithm updates the parameters of  $q^\mu$ using the SGD method to find $\mu'$ that minimises the Lagrangian, which can be treated as the loss function in training $q^\mu$ (Line 4 of Alg.~\ref{alg:train}). Then it updates $\mu$ only if the allocation vector generated by $q^{\mu'}$ and the corresponding payment vector satisfy the BF constraint (Line 5 of Alg.~\ref{alg:train}), which ensures that the trained allocation function and the corresponding payment vector always satisfy the BF constraint. The algorithm updates $\lambda$ using the gradient of the (updated) Lagrangian (Line 6 of Alg.~\ref{alg:train}). Then it increments $t$ (Line 7 of Alg.~\ref{alg:train}) and continues the training until finishing all the episodes. Finally, the algorithm outputs a trained allocation function $q^\mu$.

After training the neural-based allocation function $q^\mu$, NPQM follows the proposed framework step by step (Line 1-7 of Alg.~\ref{alg:framework}), generating the allocation vector $(q_1, \cdots, q_n) = q^\mu(\vec{b})$ and computing the expected payments $(P_1, \cdots, P_n)$, where $P_i = b_i w_i(b_i) + \int^{\overline{\theta}}_{b_i} w_i(x) dx$, $w_i(b_i) = q_i(b_i)\ln{\frac{1+q_i(b_i)}{1-q_i(b_i)}}$. The ILR $\LR_I$, illustrated in Alg.~\ref{alg:LR}, is applied to each data owner $s_i$, computing the perturbed private data $t_i'$ and the compensation $p_i$. Finally, the query is answered based on the collected dataset $\{t_1', \cdots, t_n'\}$ and the payment vector and query answer are returned. The details are illustrated in Alg.~\ref{alg:NPQM}.

%\textcolor{cyan}{(The algorithm of NPQM may be moved to the appendix to shorten the paper.)}

\begin{algorithm}
\begin{algorithmic}[1]
\caption{Neural-based Private Query Mechanism (NPQM)}
\label{alg:NPQM} \footnotesize
\Require Data owners $\vec{s}$, budget $\beta$, query $f$, allocation function $q$, training set $\vec{b}'$, learning rate $\alpha$, number of episodes $T$
\Ensure Payment vector $(p_1, \ldots, p_n)$, query answer $z$
\State  Train the allocation function $q = DualAscent(\vec{b}', \beta, \alpha, T)$
\State Generate the allocation vector $(q_1, \ldots, q_n) = q(\vec{b})$  %= q(\vec{b}, \beta)$
\State Generate the payment vector $(P_1, \ldots, P_n)$, where $P_i = b_i w_i(b_i) + \int^{\overline{\theta}}_{b_i} w_i(x) dx$
%where $p_i = b_i q_i \ln(\frac{1+q_i}{1-q_i}) + \int^{\overline{b}_i}_{b_i}q(x, b_{-i})\ln(\frac{1+q_i}{1-q_i}) dx$
%\STATE Select a subset of data owners $\vec{s}_c$ according to the allocation vector, payment vector and budget
\For{each data owner $s_i$}
\State Get a random value and the payment $t'_i, p_i = \mathcal{L}_I(t_i, q_i, P_i)$
\EndFor
\State Compute the query answer $z = f(t'_1, \ldots, t'_n)$
\State Return $(p_1, \ldots, p_n), z$ %Return $(q_1, \ldots, q_n), (p_1, \ldots, p_n), z$
\end{algorithmic}
\end{algorithm}

% NPQM trains the neural-based allocation function first and then follows the proposed framework step by step, so it possesses all the properties of the framework. During training, we only retain networks that satisfy the BF constraint, ensuring that NPQM also satisfies the BF constraint.

\begin{lemma}
    NPQM that employs $q^\mu$ as the allocation function is IC and IR.
    \label{lem:NPQM_IC_IR}
\end{lemma}

\begin{theorem}
    The neural-based private query mechanism (NPQM) is IC, IR, BF and $\varepsilon_i$-LDP.
    \label{thm:NPQM_IC_IR_BF}
\end{theorem}
%Please see the full paper 
%See Appendix C of our full paper 
See Appendix~\ref{Appendix:PRF} 
for the proof of Lemma~\ref{lem:NPQM_IC_IR} and Theorem~\ref{thm:NPQM_IC_IR_BF}.

\section{Experiments}
\subsection{Experiment setup}
In experiments, we aim to evaluate the accuracy of GPQM and NPQM when applied to various query types and datasets. The experiment setup is summarised in Table \ref{tab:setup}.

\begin{table}
    \centering
    \caption{Experiment Setups}
    \label{tab:setup} \footnotesize 
    \begin{tabular}{cc}
    \toprule
        Query $f$ & Count, median \\
        Dataset $D$ & Obesity, Maternal, Exam, Students, Salaries, Customers \\
        Budget $\beta$ & $\{0.1n, \cdots, 0.9n\}$ \\
        Bids $\vec{b}$ & Drawn from $U(0,1)$ \\
        Mechanism & GPQM (linear, log, exp), NPQM, FQ mechanisms\\
    \bottomrule
    \end{tabular}
\end{table}

\textbf{Datasets}: We use six real-world datasets, Obesity~\cite{palechor2019dataset}, Maternal Health Risk (Maternal)~\cite{ahmed2020iot}, Exam~\cite{kaggle_students_performance}, Students' Dropout and Academic Success (Students)~\cite{martins2021early}, Data Science Salaries 2023 (Salaries)~\cite{kaggle_data_science_salaries} and Customer Personality Analysis (Customers)~\cite{kaggle_customer_personality}, as the private data of data owners. See more details in App.~\ref{app:datasets}.

\textbf{Bids $\vec{b}$}: The bids of data owners are generated from the uniform distribution $U(0,1)$. In words, the range of bids $(\underline{\theta}, \overline{\theta}) = (0, 1)$.

\textbf{Budget}: The budget of consumer is set to be $0.1\overline{\theta}n$ to $0.9\overline{\theta}n$.

\textbf{Queries}: We analyse count and median queries. %{\color{magenta}Libo: Do we mention linear predictors here?} 
For count queries, we utilise the obesity level from Obesity, the risk level from Maternal, the test preparation status from Exam, the scholarship holder from Students, the remote working ratio from Salaries, and the complaint of Customers. Specifically, we count the number of overweight individuals, high-risk pregnant females, students who prepared for the exams, students who hold a scholarship, people who work remotely, and customers who made complaints.

For median queries, we use the age of Obesity and Maternal, and the reading test score of Exam, the admission test score of Students, the salary in US dollars of Salaries, and in-store shopping of Customers. Specifically, we query the median age of individuals in Obesity and Maternal, the median score in the reading test, the median salary in US dollars, and the median number of in-store shopping.

%For linear predictor queries, we use the age and exercise frequency to predict the level of obesity in the Obesity dataset, use the age and heart beat rate to predict the risk level of data owners in the Maternal dataset, and use the scores of math and writing exams to predict the score of reading exam in the Exam dataset.

\textbf{Mechanims}: We evaluate five mechanisms: three implementations of GPQM, NPQM, and FairQuery (FQ)~\cite{ghosh2011selling}. For GPQM, we consider linear, logarithmic, and exponential allocation functions, denoted by $q(\vec{b}) = 1 - \vec{b}$, $q(\vec{b}) = -\log(k_l q(\vec{b}))$, and $q(\vec{b}) = e^{-k_e q(\vec{b})}$, where the coefficients $k_l, k_e$ are randomly selected from the uniform distribution $U(0,10)$ for each query. The implementation details of NPQM are shown in Appendix~\ref{app:NPQM}.

% For NPQM, we create the training set $\vec{b}'$ by randomly sampling $n \times 50$ data points from a uniform distribution $U(0,1)$, which matches the distribution of the test set. The dual ascent algorithm is executed for a total of $T=5000$ epochs, utilising a hidden layer with size $h=8$. The initial learning rate for the stochastic gradient descent (SGD) method is set to 0.005, gradually decreasing to 0.001 towards the end. The learning rate for the Lagrangian multiplier $\lambda$ remains fixed at 0.005.

We make modifications to the FQ mechanism~\cite{ghosh2011selling} to ensure LDP for count and median queries, and we consider this modified version as our benchmark (see Appendix~\ref{app:FQ}). 
In the modified FQ, the procurement process remains unchanged such that the data owners are sorted in descending order based on the value of $b_i$, and the mechanism determines the last data owner $k$ who can be selected within the given budget. %The procurement process remains unchanged, where each data owner before $k$ is selected and compensated. 
Additionally, we introduce modifications to the query process by incorporating a widely used local randomiser {\em randomised response} \cite{warner1965randomized} that allows data owners to report their true data with a probability of $q_i$ and report a random value otherwise.

%For data owners who are not selected during the procurement process, random values are reported with no compensation. For data owners who were selected during the procurement process, probabilities of reporting the true value of private data and compensations are allocated, where the probabilities are decided based on the bids of data owners and the associated privacy loss. The local randomiser $\mathcal{L}$ is applied to each selected data owner, and the queries are issued on the reported dataset.

%The neural network $N^\mu$ used in the NPQM mechanism is trained using stochastic gradient descent optimisation over 5000 epochs, with the hidden layer size $h$ set to 8. The initial learning rate for the neural network is 0.005, which dynamically decreases to 0.001 towards the end. The learning rate for the Lagrangian multiplier $\lambda$ is fixed at 0.005.{\color{magenta}Did we mention the learning rate for neural networks before? There is only an $\alpha$ in our alg.4}

\textbf{Experiments:} For each experiment, we generate bids of data owners $\vec{b}$ randomly from the uniform distribution $U(0,1)$. We conduct $m=100$ queries for each mechanism and budget. For count queries, the performance of mechanisms is measured under average relative absolute error (RAE), i.e., $\frac{1}{m}\sum^m_{i=1} \frac{|z-z_g|}{z_g}$, where $z$ denotes the query answer generated by mechanisms and $z_g$ denotes the ground truth. For median queries, the performance is measured under average absolute error (AE), i.e., $\frac{1}{m}|z-z_g|$.

\textbf{Implementation details:} NPQM is implemented in Python 3.9 on NVIDIA GeForce RTX 3090 Ti GPU. All mechanisms are implemented in Python 3.9 on Apple M1 Pro CPU. The code is available on {\tt \url{ https://anonymous.4open.science/r/IntegratedPDQS}}

%{\color{magenta}Do we write or delete this section?}

\subsection{Results and discussion}

\begin{figure}
    \centering
    \includegraphics[width=0.95\linewidth]{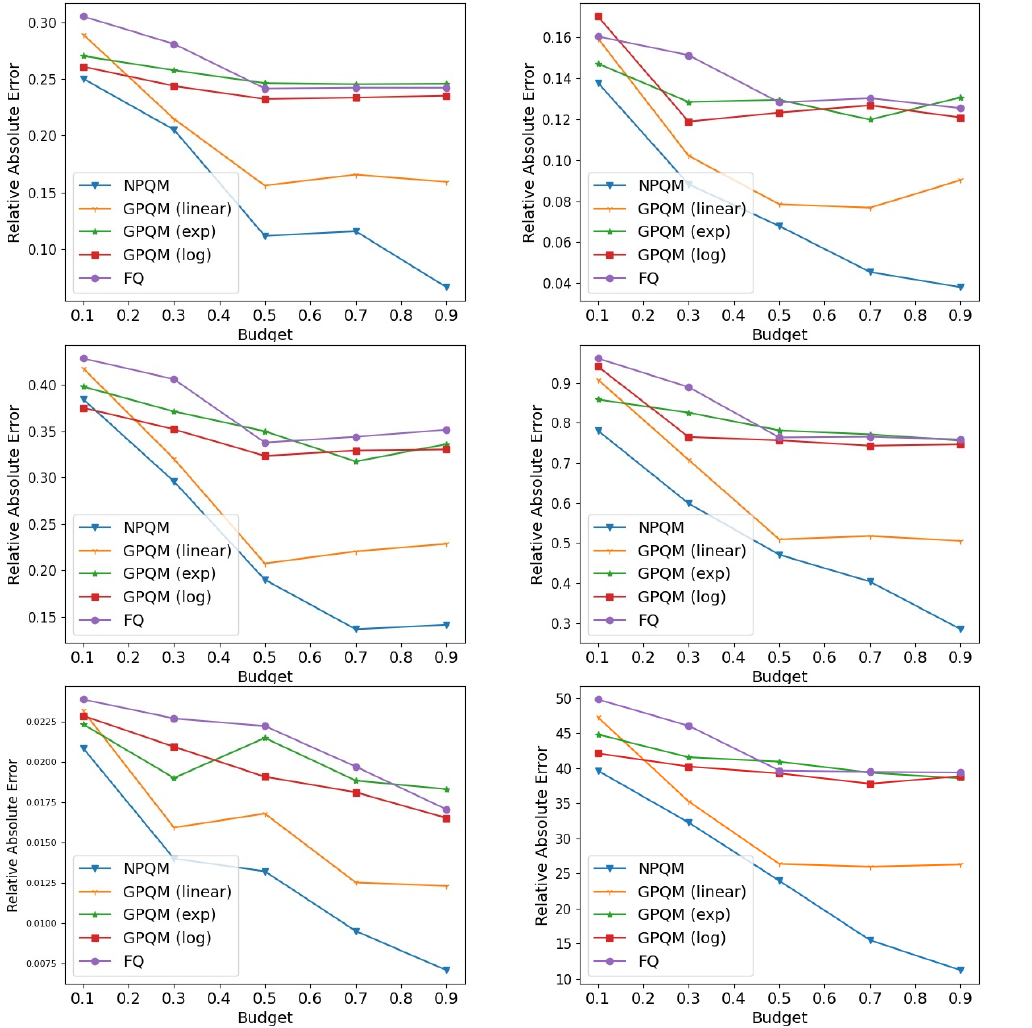}
    \caption{Relative absolute error of count queries on the Obesity (top-left), Maternal (top-right), Exam (middle-left), Students (middle-right), Salaries (bottom-left) and Customers (bottom-right) datasets}
    %{\red (M: the text in the figures are too small) }}
    \label{exp:accurayCount}
    %\vspace{-2.2em}
\end{figure}

\begin{figure}
    \centering
    \includegraphics[width=0.95\linewidth]{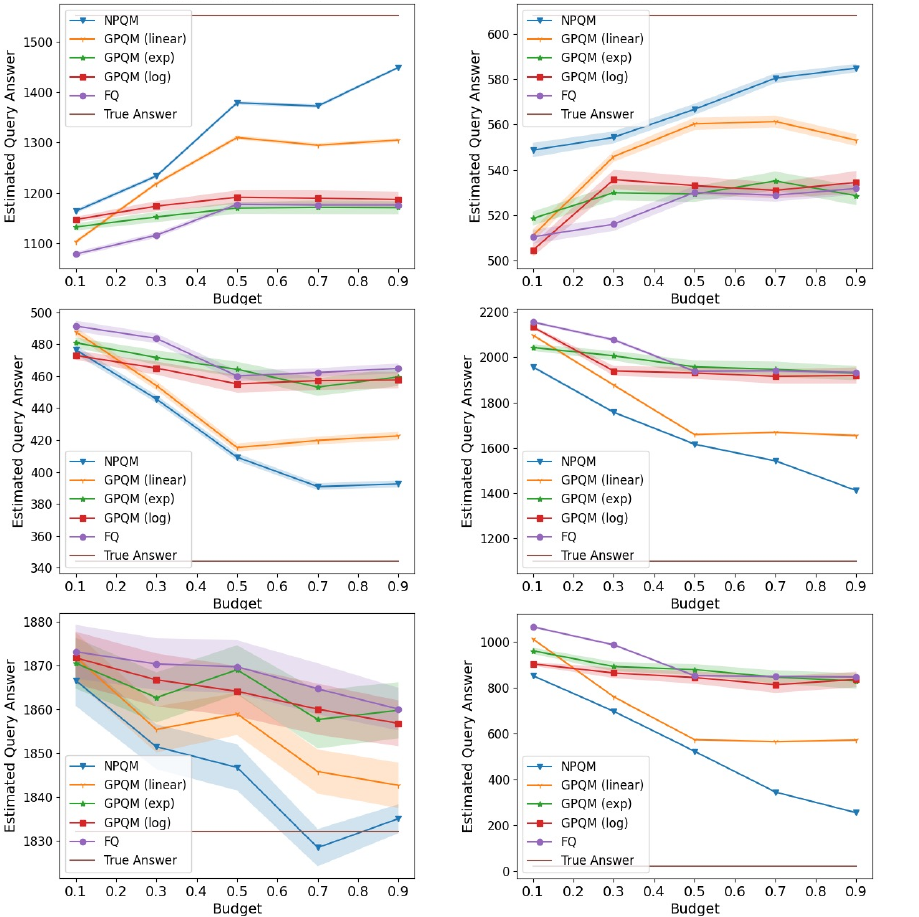}
    \caption{Confidence interval of the estimated answer of count queries on the Obesity (top-left), Maternal (top-right), Exam (middle-left), Students (middle-right), Salaries (bottom-left) and Customers (bottom-right) datasets}
    \label{exp:CIcount}
\end{figure}

We evaluate and compare the performance of the five mechanisms. Fig~\ref{exp:accurayCount}-\ref{exp:CIcount} display the RAE and confidence intervals of the mechanisms for count queries on the six datasets, respectively. %The AE and confidence intervals for median queries are detailed in the full paper. 
Fig~\ref{exp:accurayMedian}-\ref{exp: CI median} illustrate the AE and confidence intervals of the mechanisms for median queries on the six datasets, respectively. 

In most experiments, our approach outperforms the benchmark. NPQM achieves the best accuracy on both count and median queries, with the lowest RAE / AE and the estimated query answer closest to the ground truth. GPQM with a linear allocation function also performs well. GPQM with exponential and logarithmic allocation functions outperforms the benchmark in most experiments when the budget is below $0.5n$, but this advantage becomes less noticeable when the budget exceeds $0.5n$.

As expected, the performance of NPQM and GPQM with linear allocation improves as the budget increases. The advantages of NPQM and GPQM with linear allocation become more significant compared to other mechanisms as the budget increases, indicating that these mechanisms can approximate the ground truth with a sufficient budget. %For example, in the case of performing median queries on the Customers dataset with a budget of $0.7n$, NPQM achieves RAE and AE close to 0, and the estimated query answer matches the ground truth. 
On the other hand, for GPQM with exponential and logarithmic allocation and FQ, their accuracy does not improve further as the budget reaches $0.5n$. 

%The experimental results demonstrate that NPQM consistently performs better and can provide accurate estimates for various query types and datasets. On the other hand, GPQM's performance is influenced by the choice of allocation function. When an appropriate allocation function is used, GPQM performs well, but its performance is limited when the allocation function is not suitable. As for our benchmark, it suffers from limitations as it needs to allocate a portion of the budget to data owners who cannot truthfully report their private data, resulting in compromised accuracy even with a sufficient budget.

The experimental results demonstrate that NPQM consistently performs better and provides accurate estimates for various query types and datasets. NPQM also performs well when bids follow a normal distribution, and even when the training and test sets come from different distributions (see Appendix~\ref{app:exp result}). Despite the potential budget overruns from using expected payment as a training constraint, NPQM remains an accurate solution, as such occurrences are rare and can be resolved through system re-runs. % More experiments show that NPQM achieves higher accuracy than the benchmark, even when data owner bids follow normal distributions and when the training set and test set belong to different distributions. Please refer to the full paper for details.
Besides, we observe that GPQM's performance relies on the choice of the allocation function. As for our benchmark, it suffers from limitations as it needs to allocate a portion of the budget to data owners who cannot truthfully report their private data, resulting in compromised accuracy even with a sufficient budget.

\section{Conclusion}
We introduce an integrated PDQS framework, which combines the procurement and query processes, and effectively utilises the consumer's budget to approximate query accuracy under LDP. We propose two implementations of the novel framework, GPQM and NPQM, which address queries while considering IC, IR, and BF constraints. The experimental results demonstrate that our mechanisms outperform existing approaches that separate the procurement and query processes in query accuracy. Potential future work can be extending the Integrated PDQS framework to handle sequential data and multi-dimensional private data with varying privacy requirements for different dimensions.

%As potential future work, the Integrated PDQS framework can be extended to handle different types of private data such as sequential data, and the scenario where the private data is multi-dimensional with various privacy requirements for different dimensions.

\ack This work is partially supported by National Natural Science Foundation of China No. 62172077.

\bibliography{ecai.bib}

\clearpage

\appendix
\section{FairQuery}
\label{app:FQ}

\begin{algorithm}
\begin{algorithmic}[1]
\caption{FairQuery of local DP}
\label{alg:FQLDP}
\Require Data owners $\vec{s}$, Budget $\beta$, Query $f$
\Ensure Allocation vector $(q_1, \ldots, q_n)$, 
Payment vector $(p_1, \ldots, p_n)$, Query answer $z$
\State Sort data owners such that $b_1 \leq b_2 \leq \cdots \leq b_n$
\State Let $k$ be the largest integer such that $\frac{b_k}{n-k} \leq \frac{\beta}{k}$
\State Pay each $i>k$ $p_i=0$ and each $i \leq k$ $p_i = \min(\frac{\beta}{k}, \frac{b_{k+1}}{n-k})$
\For{$i \leq k$}
\State Set $q_i = \frac{e^{\varepsilon}-1}{e^{\varepsilon} + 1}$, where $\varepsilon = \frac{1}{n-k}$
\State Report $t_i' = t_i$ with probability $q_i$, report a random value $t_i' \in \mathcal{D}$ otherwise
\EndFor
\For{$i > k$}
\State Report a random value $t_i' \in \mathcal{D}$
\EndFor
\State Return $(p_1, \ldots, p_n), z$
\end{algorithmic}
\end{algorithm}

\section{QMIX}
\label{appendix: QMIX}
%We explain the details of QMIX.
QMIX considers cooperation in a multi-agent problem, in which the single agent's estimation of benefit should conform to the team's benefit. Intuitively, QMIX applies a function that takes each agent's estimation as input, and calculates the team's estimation as output. Thus requiring the monotonicity of the function can guarantee the conformity. 

The cooperative environment is defined as a Dec-POMDP $G=(S,U,P,O,Z,r,n,\gamma)$. $s \in S$ describes the state. Each agent $a \in A$ chooses an action $u^a \in U^a$ and forms the joint action $u \in U$. The dynamic of the system is controlled by the transition function $P(s'|s,u):S \times S \times U \rightarrow [0,1]$. $r(s,u): S \times U \rightarrow \mathbf{R}$ describes the reward of all agents, and $\gamma \in [0,1)$ is a discount factor. Here we consider the partial observable situation, in which each agent can only observe part of the environment state $z \in Z$ according to the function $O(s,a): S \times A \rightarrow Z$. Each agent also maintains an observation history $\tau^a \in T \equiv (Z \times U)^*$, and the utility function $Q_a(\tau,u_t)=\mathbf{E}_{s_{t+1}:\infty,u_{t+1}:\infty[R_t|\tau,u_t]}$ that estimates the expectation of future rewards, and $R_t=\sum_{i=0}^\infty \gamma^ir_{t+i}$ is the discounted reward. Similarly, we want a $Q_{tot}(\tau,u_t)$ as a global utility function. So, the cooperation is considered as a monotonicity between $Q_{tot}$ and $Q_a$: $\frac{\partial Q_{tot}}{\partial Q_a} \geq 0, \forall a \in A$. The monotonicity is achieved by a neural network that takes $Q_a$ as input, and $Q_{tot}$ as output. Since the network guarantees monotonicity, which is one condition in our framework as well. We apply a similar neural network, and will illustrate our neural network below.

\section{Proof of Lemma~\ref{lem:NPQM_IC_IR} and Theorem~\ref{thm:NPQM_IC_IR_BF}}\label{Appendix:PRF}

%\begin{lemma}
{\bf Lemma~\ref{lem:NPQM_IC_IR}} {\it  NPQM that employs $q^\mu$ as the allocation function is IC and IR.}  
%\end{lemma}

\medskip

\begin{proof}
Similar to Lemma~\ref{lem:GPQM_IC_IR}, we prove lemma~\ref{lem:NPQM_IC_IR} by Archer and Tardos' theorem~\cite{archer2001truthful}.  The monotonicity of the function $q^\mu$ is guaranteed by enforcing the negation of the absolute value of $w_1$ and the absolute value of $w_2$, and the expected payment $P_i$ is structured in the same form as in GPQM. Thus, Conditions 1 and 3 are met.
The range of the function is guaranteed to be $(0,1)$, by applying the Sigmoid function, which is a common normalisation technique that maps any inputs to values between 0 and 1. Similar to the previous proof of lemma~\ref{lem:GPQM_IC_IR}, we can show that the integral is finite and Condition 2 is also met.
\end{proof}

\medskip

%\begin{theorem}
\noindent {\bf Theorem~\ref{thm:NPQM_IC_IR_BF}}
{\it The neural-based private query mechanism (NPQM) is IC, IR, BF and $\varepsilon_i$-LDP.}
%\end{theorem}

\medskip

\begin{proof} %[Theorem~\ref{thm:NPQM_IC_IR_BF}]
The training algorithm employed by NPQM updates the parameter $\mu$  of the allocation function only if the corresponding allocation probabilities and expected payments satisfy the BF constraint, so the BF constraint is guaranteed by the construction of NPQM. Also, NPQM employs the ILR $\LR_I$ that ensures $\varepsilon_i$-LDP. Together with lemma~\ref{lem:NPQM_IC_IR}, NPQM is IC, IR, BF and $\varepsilon_i$-LDP.
\end{proof}

\section{Experiment Details}

\subsection{Dataset Description}
\label{app:datasets}

The Obesity dataset comprises 2111 records of individuals from South America, detailing their dietary habits, physical condition, and estimated levels of obesity. The Maternal dataset documents the physical condition of 1014 pregnant women, including data such as age, blood pressure, and blood sugar, which can impact maternal mortality. The Exam dataset contains personal information and exam performance of 1000 American high school students. The Students dataset provides information on the personal details and academic performance of 4424 undergraduate students. The Salaries dataset records the personal information and income of 3755 individuals working in the field of data science. The Customers dataset describes the personal details and shopping information of 2240 customers.

\subsection{NPQM Settings}
\label{app:NPQM}
We generate the bids in the training and test sets for NPQM.
We use four different setups to verify the robustness of our NPQM. {\bf (1).} Both the training and test sets are randomly sampled from the uniform distribution $U(0,1)$. {\bf (2).} The training set is drawn from $U(0,1)$, while the test set are drawn from a normal distribution (std=0.25) with all values mapped to $(0, 1)$. {\bf (3).} The training and test sets are sampled from the standard normal distribution with all values mapped to $(0, 1)$. {\bf (4).} The training set is generated as in Experiment (3), while the samples of the test set are drawn from a normal distribution (std=0.25) with all values mapped to $(0, 1)$.

%\textcolor{cyan}{For experiment 1 (as shown in Fig.~\ref{exp:accurayCount}-\ref{exp: CI median}), we create the training set $\vec{b}$ by randomly sampling $n \times 50$ data points from a uniform distribution $U(0,1)$, which matches the distribution of the test set. For experiment 2 (as shown in Fig.~\ref{exp: normal_uniform_count}-\ref{exp: normal_uniform_median}), we create the training set as in the previous experiment, but the samples for the test set are drawn from a normal distribution (std=0.25) with all values mapped to $(0, 1)$. For experiment 3 (as shown in Fig.~\ref{exp: normal_count}-\ref{exp: normal_median}), we create both the training and test sets by randomly sampling $n \times 50$ data points from a standard normal distribution with all values mapped to $(0, 1)$. For experiment 4 (as shown in Fig.~\ref{exp:diff_normal_count}-\ref{exp:diff_normal_median}), we create the training set as in the previous experiment, but the samples for the test set are drawn from a normal distribution (std=0.25) with all values mapped to $(0, 1)$.}

%For NPQM, we create the training set $\vec{b}'$ by randomly sampling $n \times 50$ data points from a uniform distribution $U(0,1)$, which matches the distribution of the test set. 
The dual ascent algorithm is executed for a total of $T=5000$ epochs, utilising a hidden layer with size $h=8$. The initial learning rate for the stochastic gradient descent (SGD) method is set to 0.005, gradually decreasing to $0.001$ towards the end. The learning rate for the Lagrangian multiplier $\lambda$ remains fixed at $0.005$.

\section{Experiment Results}
\label{app:exp result}

Fig.~\ref{exp:accurayMedian}-\ref{exp: CI median} demonstrate the AE and confidence intervals of the mechanisms for median queries on the six datasets, respectively. 

Figure~\ref{exp:CIcount}-\ref{exp:diff_normal} demonstrate the results for different setups regarding the training and test sets. Figures~\ref{exp:CIcount}-\ref{exp: CI median} demonstrate the accuracy for Experiment (1), while Figures~\ref{exp:normal_uniform}, \ref{exp:normal} and \ref{exp:diff_normal} demonstrate the accuracy of NPQM and FQ mechanisms on Obesity and Exam datasets for Experiment (2), (3) and (4), respectively.
%Fig.~\ref{exp: normal_uniform_count}-\ref{exp: normal_uniform_median} demonstrate the accuracy of NPQM and FQ mechanisms on Obesity and Exam datasets where the training set follow uniform distribution while the test set comes from a normal distribution (std=0.25). Fig.~\ref{exp: normal_count}-\ref{exp: normal_median} demonstrate the accuracy of NPQM and FQ mechanisms on Obesity and Exam datasets where the training and test sets come from the standard normal distribution. Fig.~\ref{exp:diff_normal_count}-\ref{exp:diff_normal_median} demonstrate the accuracy of NPQM and FQ mechanisms on Obesity and Exam datasets where the training and test sets come from different normal distributions.

\begin{figure}
    \centering
\includegraphics[width=0.95\linewidth]{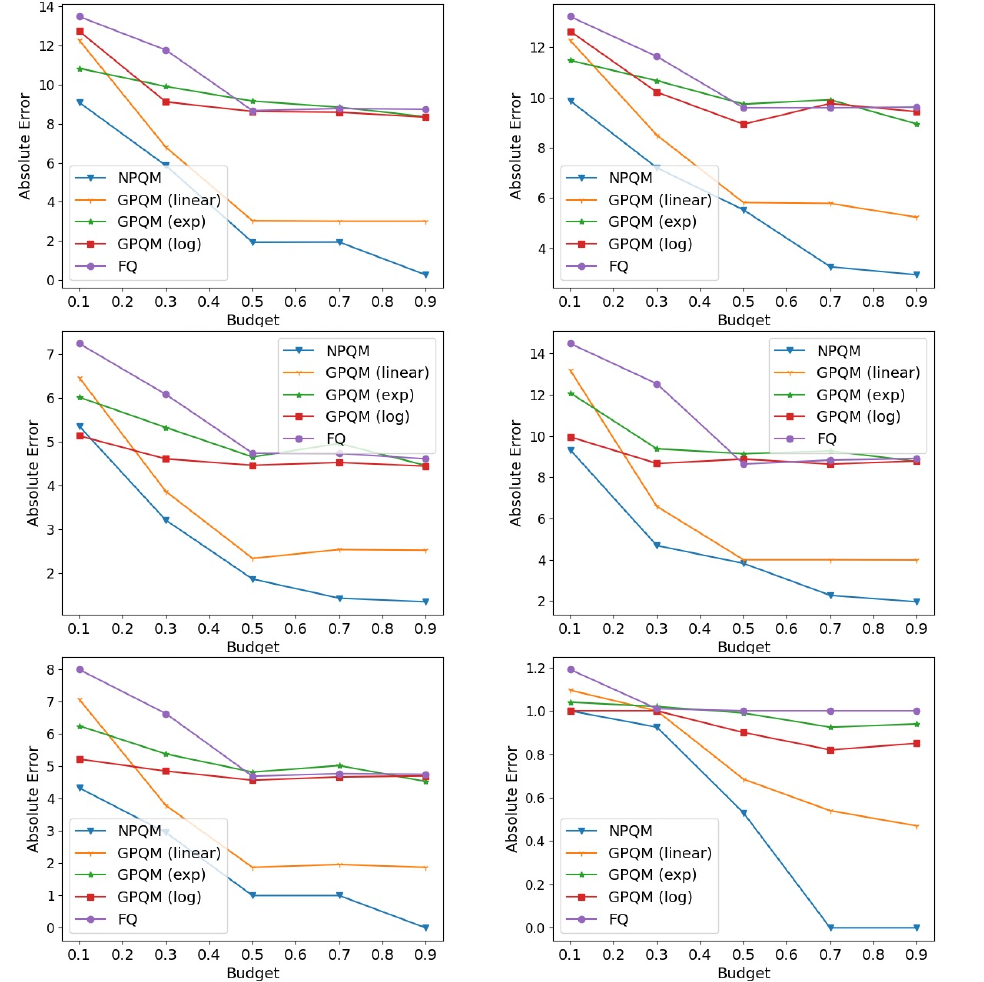}
    \caption{Absolute error of median queries on the Obesity (top-left), Maternal (top-right), Exam (middle-left), Students (middle-right), Salaries (bottom-left) and Customers (bottom-right) datasets}
    \label{exp:accurayMedian}
\end{figure}

\begin{figure}
    \centering
\includegraphics[width=0.95\linewidth]{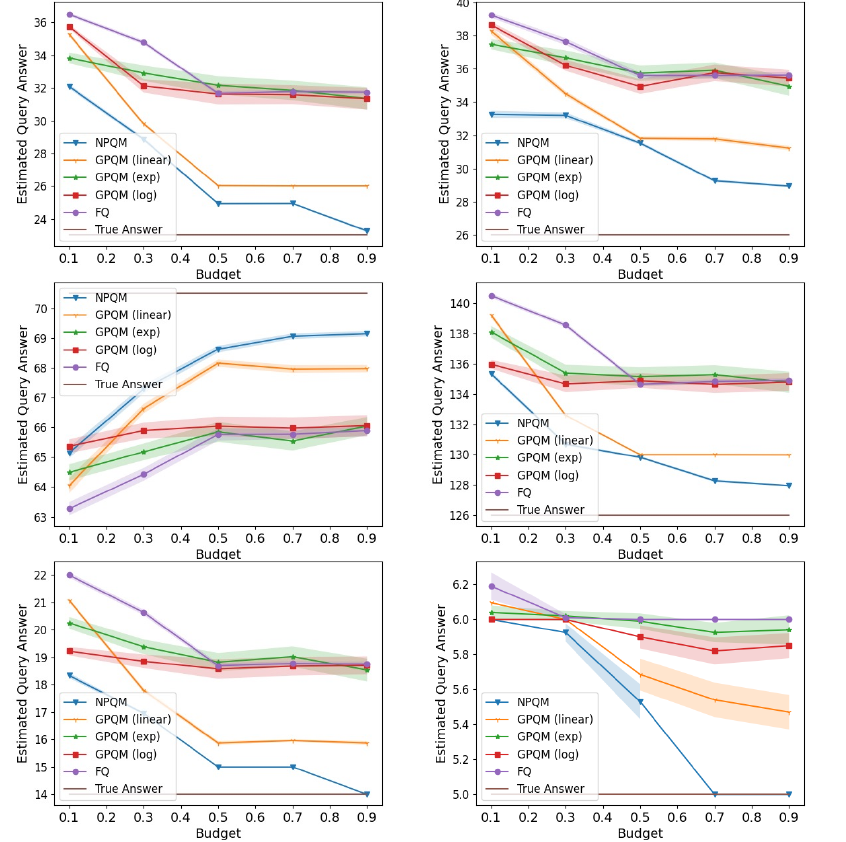}
    \caption{Confidence interval of the estimated answer of median queries on the Obesity (top-left), Maternal (top-right), Exam (middle-left), Students (middle-right), Salaries (bottom-left) and Customers (bottom-right) datasets}
    \label{exp: CI median}
\end{figure}

\begin{figure}
    \centering
\includegraphics[width=0.95\linewidth]{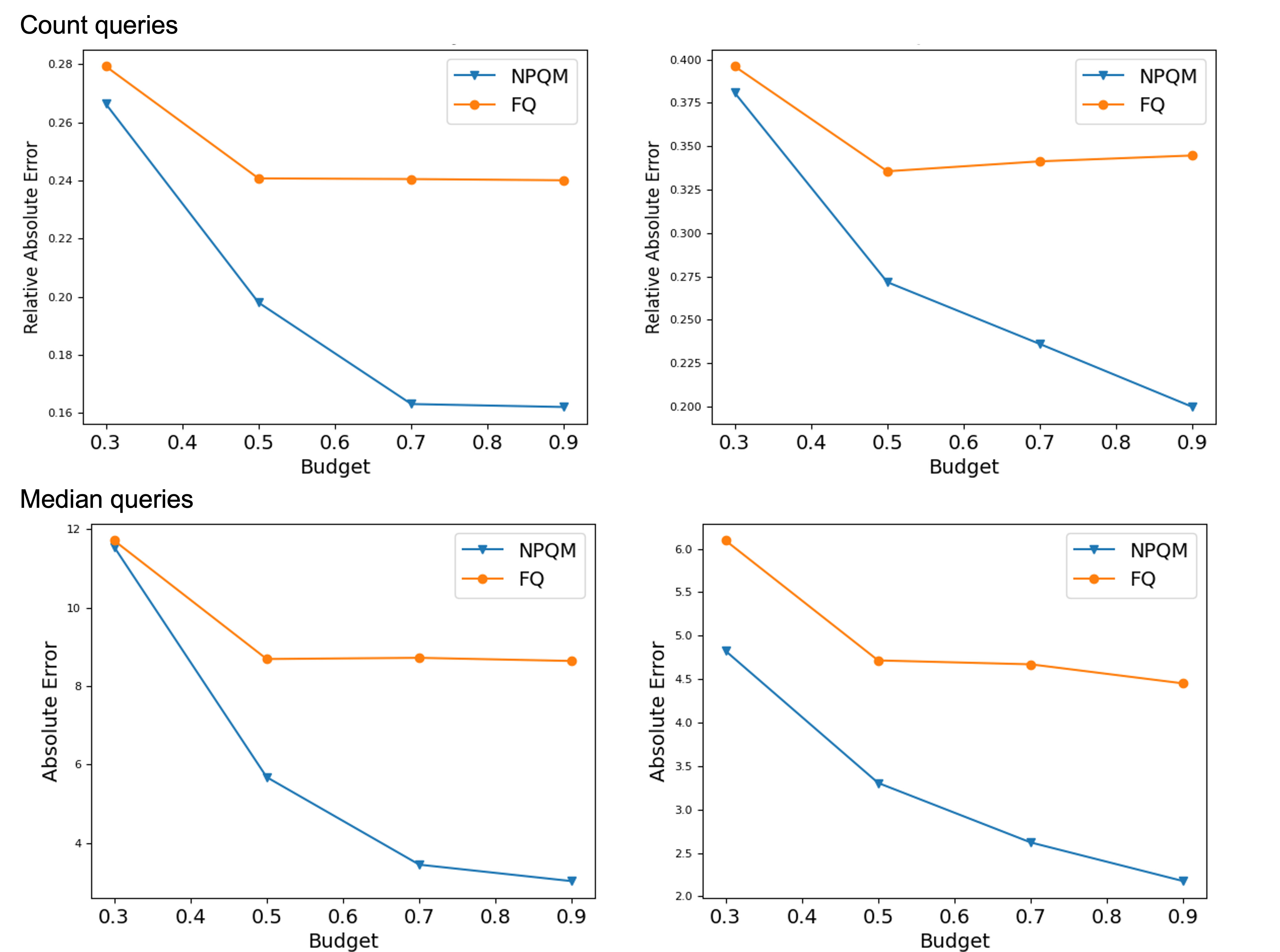}
    \caption{Relative absolute error of count and median queries on the Obesity (left) and Exam (right) datasets, where the training set of NPQM is under uniform distribution while the test set follows a normal distribution (std=0.25) with all values mapped to the range of 0 to 1}
    \label{exp:normal_uniform}
\end{figure}

\begin{figure}
    \centering
\includegraphics[width=0.95\linewidth]{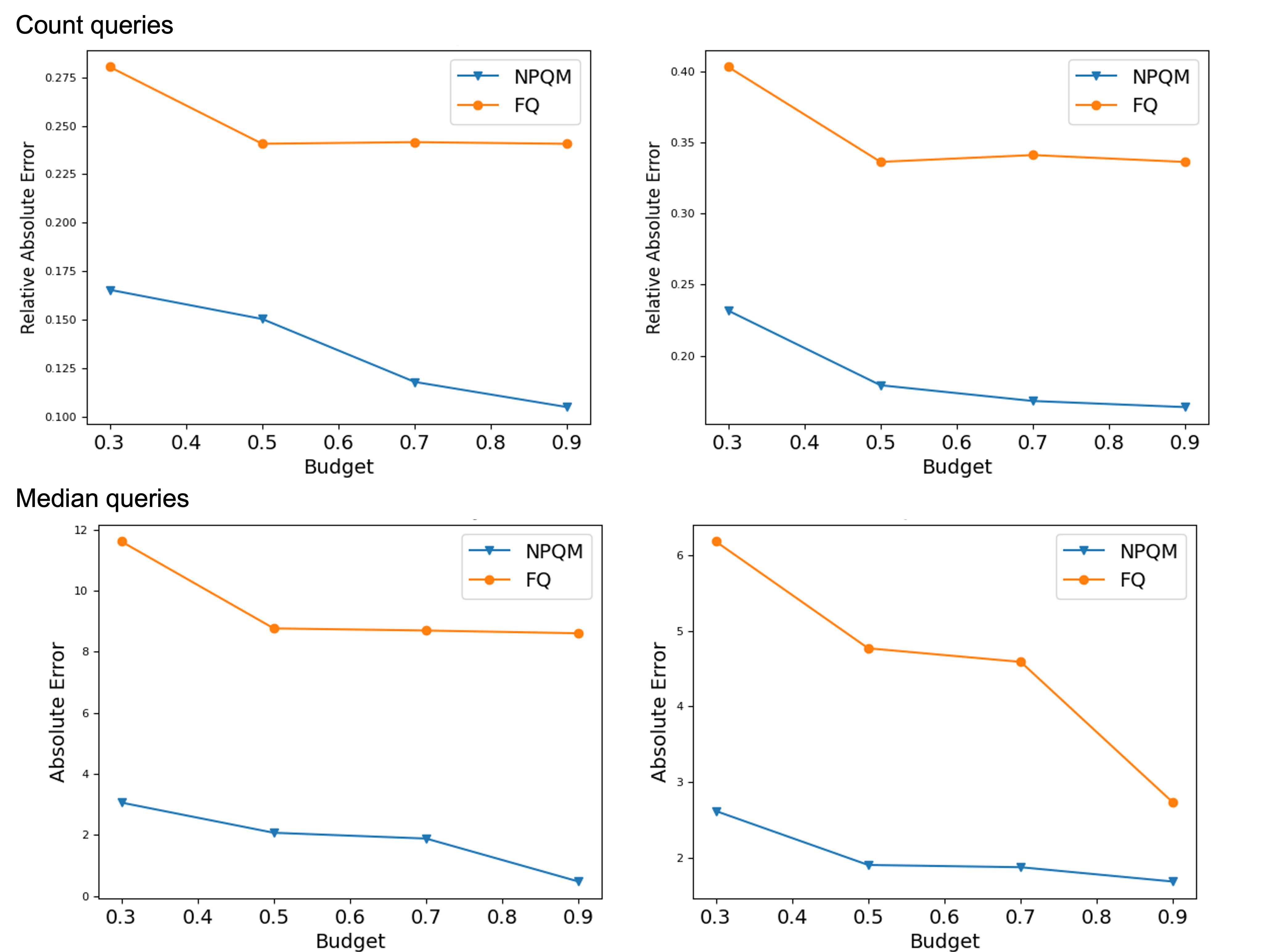}
    \caption{Relative absolute error of count and median queries on the Obesity (left) and Exam (right) datasets, where the training and test sets follow a standard normal distribution with all values mapped to the range of 0 to 1}
    \label{exp:normal}
\end{figure}

%\newpage

\begin{figure}
    \centering
\includegraphics[width=0.95\linewidth]{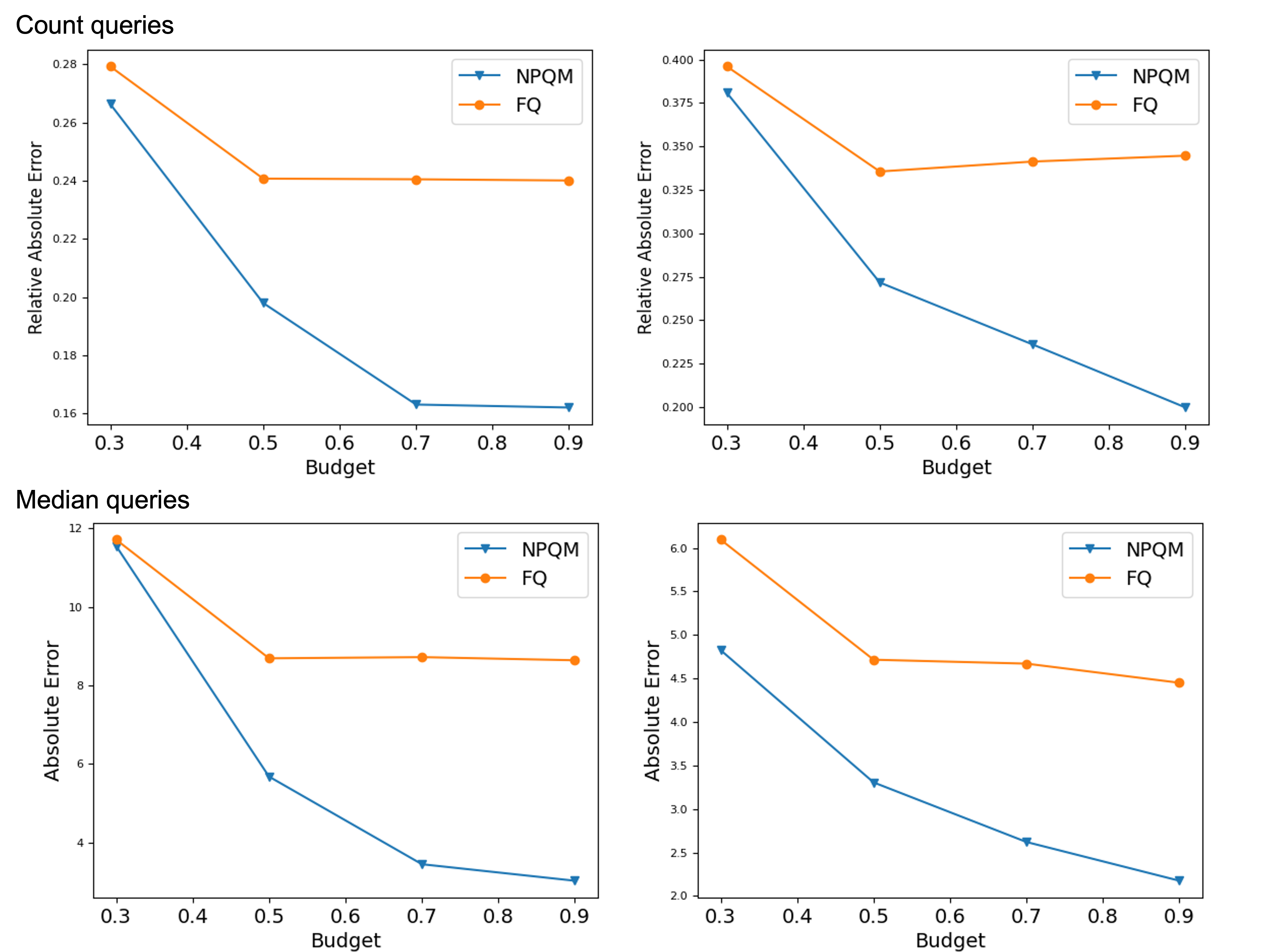}
    \caption{Relative absolute error of count and median queries on the Obesity (left) and Exam (right) datasets, where the training set is under a standard normal distribution and the test set follows a normal distribution (std=0.25) with all values mapped to the range of 0 to 1}
    \label{exp:diff_normal}
\end{figure}

\end{document}